\documentclass[11pt,letterpaper]{article}
\usepackage[margin=1in]{geometry}
\usepackage{amsmath,amsfonts,amsthm,amssymb,color}
\usepackage[usenames,dvipsnames,svgnames,table]{xcolor}
\definecolor{darkgreen}{rgb}{0.0,0,0.9}
\usepackage{mathtools}
\usepackage{comment}
\usepackage{tikz}
\usepackage{bbm}
\usepackage{dsfont}
\usepackage[sc]{mathpazo}
\usepackage[basic]{complexity}
\usepackage{algorithm2e}
\usepackage[colorlinks=true,citecolor=OliveGreen,linkcolor=BrickRed,urlcolor=BrickRed,pdfstartview=FitH]{hyperref}
\usepackage[capitalize,nameinlink]{cleveref}
\usepackage{tcolorbox}

\newtcolorbox{wbox}
{
	colback  = white,
}

\usepackage{csquotes}
\usepackage[short]{optidef}
\usetikzlibrary{shapes, arrows}
\newtheorem{theorem}{Theorem}
\newtheorem{remark}[theorem]{Remark}

\newtheorem{alg}[theorem]{Mechanism}

\newtheorem{lemma}[theorem]{Lemma}
\newtheorem{corollary}[theorem]{Corollary}
\newtheorem{proposition}[theorem]{Proposition}

\theoremstyle{definition}
\newtheorem{definition}[theorem]{Definition}

\SetKwInOut{Input}{Input}
\SetKwInOut{Output}{Output}
\SetKwFunction{Uncross}{\textsc{Uncross}}
\SetKwFunction{MergeUncross}{\textsc{MergeUncross}}
\SetKwFunction{PerfectMatching}{\textsc{PerfectMatching}}
\SetKwBlock{InParallel}{in parallel do}{end}
\SetKwFor{ParallelFor}{for}{in parallel do}{end}

\let\R\relax
\newcommand*{\R}{\mathbb{R}}

\newcommand*{\Qplus}{\mathbb{Q_+}}

\newcommand*{\suppress}[1]{}

\newcommand*{\cT}{\mathcal{T}}

\newcommand{\bbR}{\mathbb{R}}

\makeatletter
\def\thm@space@setup{%
	\thm@preskip= 10pt
	\thm@postskip=\thm@preskip 
}
\makeatother

\makeatletter
\renewcommand{\paragraph}{%
	\@startsection{paragraph}{4}%
	{\z@}{5pt}{-1em}%
	{\normalfont\normalsize\bfseries}%
}
\makeatother
\newenvironment{fminipage}%
{\begin{Sbox}\begin{minipage}}%
		{\end{minipage}\end{Sbox}\fbox{\TheSbox}}

\newcommand\vv{\boldsymbol{\mathit{v}}}

\newcommand{\CN}{\mbox{${\mathcal N}$}}
\newcommand{\CS}{\mbox{${\mathcal S}$}}
\newcommand{\CM}{\mbox{${\mathcal M}$}}

\title{One-Sided Matching Markets with Endowments:\\ Equilibria and Algorithms}

\author{Jugal Garg\thanks{University of Illinois at Urbana-Champaign. Supported by NSF Grant CCF-1942321 (CAREER).}\\ \texttt{\small jugal@illinois.edu} \and Thorben Tr\"obst\thanks{University of California, Irvine. Supported in part by NSF grant CCF-1815901.}\\ \texttt{\small t.troebst@uci.edu} \and Vijay V. Vazirani\thanks{University of California, Irvine. Supported in part by NSF grant CCF-1815901.}\\ \texttt{\small vazirani@ics.uci.edu}}
\date{}

\begin{document}
\maketitle

\begin{abstract}
	The Arrow-Debreu extension of the classic Hylland-Zeckhauser scheme \cite{HZ79} for a one-sided matching market -- called ADHZ in this paper -- has natural applications but has instances which do not admit  equilibria. By introducing approximation, we define the {\em $\epsilon$-approximate ADHZ model}, and we give the following results.

	\begin{enumerate}
		\item Existence of equilibrium under linear  utility functions. We prove that the equilibrium satisfies Pareto optimality, approximate envy-freeness, and approximate weak core stability.
		\item A combinatorial polynomial time algorithm for an $\epsilon$-approximate ADHZ equilibrium for the case of dichotomous, and more generally bi-valued, utilities. 
		\item An instance of ADHZ, with dichotomous utilities and a strongly connected demand graph, which does not admit an equilibrium.
			\end{enumerate}
			
	Since computing an equilibrium for HZ is likely to be highly intractable \cite{VY-HZ} and because of the difficulty of extending HZ to more general utility functions, \cite{HV-NBMM} proposed (a rich collection of) Nash-bargaining-based matching market models. For the dichotomous-utilities case of their model {\em linear Arrow-Debreu Nash bargaining one-sided matching market (1LAD)}, we give a combinatorial, strongly polynomial time algorithm and show that it admits a rational convex program. 

\end{abstract}

\section{Introduction}\label{sec.intro}
In this paper, we define an Arrow-Debreu extension of the classic Hylland-Zeckhauser (HZ) mechanism~\cite{HZ79} for one-sided matching markets. This fills a void in the space of general\footnote{As opposed to mechanisms for specific one-sided matching markets.} mechanisms for one-sided matching markets. Such mechanisms are classified according to two criteria: whether they use cardinal or ordinal utility functions, and whether they are in the Fisher or Arrow-Debreu\footnote{This is also called the Walrasian or exchange setting.} setting. The other three possibilities are covered as follows: (cardinal, Fisher) by the Hylland-Zeckhauser scheme \cite{HZ79}; (ordinal, Fisher) by Probabilistic Serial \cite{Bogomolnaia-PS} and Random Priority  \cite{Moulin2018fair}; and (ordinal, Arrow-Debreu) by Top Trading Cycles \cite{Shapley1974-TTC}. Details about these mechanisms are given in Section \ref{sec.related} and Appendix~\ref{sec.ord-car}. 

The Arrow-Debreu setting of one-sided matching markets has several natural applications beyond the Fisher setting, e.g., allocating students to rooms in a dorm for the next academic year, assuming their current room is their initial endowment. The issue of obtaining such an extension of the HZ mechanism, called {\em ADHZ} in this paper,  was studied by Hylland and Zeckhauser. However, this culminated in an example which inherently does not admit an equilibrium \cite{HZ79}. 

One recourse to this was given by Echenique, Miralles and Zhang \cite{Echenique2019constrained} via their notion of an {\em $\alpha$-slack Walrasian equilibrium}: This is a hybrid between the Fisher and Arrow-Debreu settings. Agents have initial endowments of goods and for a fixed $\alpha \in (0, 1]$, the budget of each agent, for given prices of goods, is $\alpha + (1 - \alpha) \cdot m$, where $m$ is the value for her initial endowment; the agent spends this budget to obtain an optimal bundle of goods. Via a non-trivial proof, using the Kakutani Fixed Point Theorem, they proved that an $\alpha$-slack equilibrium always exists. 

In this paper, we show that we can remain with a pure Arrow-Debreu setting provided we relax the notion of equilibrium to an {\em approximate equilibrium}, a notion that has become common-place in the study of equilibria within computer science. We call this the {\em $\epsilon$-approximate ADHZ model}. For this model, we give the following results.

{\bf 1).} We prove the existence of an equilibrium for arbitrary cardinal utility functions, using the fact from \cite{Echenique2019constrained} that an $\alpha$-slack equilibrium always exists for $\alpha > 0$. 

{\bf 2).} We prove that the equilibrium in our $\epsilon$-approximate ADHZ model is Pareto optimal, approximately envy free, and approximately weak core stable. In contrast, the allocation found by an HZ equilibrium is Pareto optimal and envy-free \cite{HZ79} and incentive compatible in the large \cite{He2018pseudo}. 

{\bf 3).} For an Arrow-Debreu market under linear utilities, Gale~\cite{Gale76} defined a {\em demand graph:} a directed graph on agents with an edge $(i, j)$ if agent $i$ likes a good that agent $j$ has in her initial endowment. He proved that a sufficiency condition for the existence of equilibrium is that this graph be strongly connected. The following question arises naturally: Is this a sufficiency condition for equilibrium existence in ADHZ as well? We provide a negative answer to this question. We give an instance of ADHZ whose demand graph is not only strongly connected but also has dichotomous utilities, and yet it does not admit an equilibrium. 

{\bf 4).} For the case of dichotomous utilities, we give a combinatorial polynomial-time algorithm for computing an equilibrium for our $\epsilon$-approximate ADHZ model. This result also extends to the case of bi-valued utilities, i.e., each agent's utility for individual goods comes from a set of cardinality two, though the sets may be different for different agents. We note that the polynomial-time algorithm of~\cite{DM15,DuanGM16} for Arrow-Debreu markets under linear utilities, as well as the recent strongly polynomial-time algorithm for the same problem~\cite{GV19} are quite complicated, in particular because they resort to the use of balanced flows, which uses the $l_2$ norm. In contrast, we managed to avoid the use of $l_2$ norm and hence we obtain a simple algorithm.

 A corollary of the last result is that the equilibrium of the dichotomous utilities case of the $\epsilon$-approximate ADHZ model involves only rational numbers. In contrast we give an instance of ADHZ whose unique equilibrium has irrational prices and allocations. This instance is obtained by appropriately modifying an instance for the HZ model, given in \cite{VY-HZ}, whose (unique) equilibrium has irrational prices and allocations.  

\bigskip

\noindent{\bf Computational complexity of HZ, and the Nash bargaining alternative:} Recently \cite{VY-HZ} gave the first comprehensive study of the computational complexity of HZ. Their main message was that it is likely to be highly intractable; however, they left open the question of giving a formal proof of this fact via a hardness result, see Section \ref{sec.related}. From the viewpoint of use in applications, the intractability of HZ is a serious drawback. As noted in \cite{HV-NBMM}, solving even small instances with $n = 4$ or 5 is difficult. Furthermore, in view of the above-stated difficulties of obtaining an Arrow-Debreu extension of the HZ, extensions to more general utility functions or to two-sided matching markets were never considered. 

To deal with this reality, \cite{HV-NBMM} defined a rich collection of Nash-bargaining-based matching market models, both one-sided and two-sided, in both Fisher and Arrow-Debreu settings, together with implementations, using available solvers, and experimental results. Encouraged by their results, \cite{Panageas2021combinatorial} gave efficient combinatorial algorithms for several of these models. Additionally, they established a deep connection between HZ and the Nash-bargaining-based models, hence confirming that the alternative to HZ proposed in \cite{HV-NBMM} is a principled one. The Nash bargaining solution has very desirable properties: it is Pareto optimal and symmetric and has been found to be remarkably fair, e.g., see Remark \ref{rem.remarks} and \cite{Nash-Unreasonable, Abebe-MM-Truthful, Moulin2018fair}. This aspect  has been further explored under the name of Nash Social Welfare \cite{cole2018approximating, cole2017convex}.

Continuing on this line of work, we explore the dichotomous-utilities case of the model {\em linear Arrow-Debreu Nash bargaining one-sided matching market (1LAD)} defined in \cite{HV-NBMM}. Let us call this case {\em 1DLAD}. Optimal allocations to an instance of {\em 1DLAD} are obtained by optimally solving the non-linear convex program \eqref{eq.CP-LAD-copy}. We ask if 1DLAD is polynomial time solvable. A prerequisite for this is that each instance of 1DLAD should admit a rational-valued equilibrium; we establish this by showing that \eqref{eq.CP-LAD-copy} is a \emph{rational convex program (RCP)} for the case of dichotomous utilities; see Appendix~\ref{sec.RCP1} for a definition. 

Our proof of rationality provides valuable insights: It turns out that the dual of \eqref{eq.CP-LAD-copy} has two types of variables, one corresponding to goods and the other corresponding to agents; these play the roles of {\em prices of goods} and {\em price-offsets for agents}. Additionally, it indicates how the money, $m_i$ of an agent $i$ should be defined and exactly how $i$'s allocation needs to be paid for. Using these insights, we give a novel market, based on the Fisher model, whose equilibrium captures an optimal solution to \eqref{eq.CP-LAD-copy}. We call it the {\em variable-budget market model}. We next give a combinatorial, strongly polynomial-time mechanism for computing an equilibrium for this model. It turns out that $m_i$ is a function of the eventual utility, $v_i$ of agent $i$. Our mechanism iteratively updates $v_i$, and as a result, it also keeps updating $m_i$, hence the name of the model. 

\bigskip

\noindent {\bf Additional results in the Appendix:} 
The paper \cite{VY-HZ} had given a combinatorial polynomial time algorithm for HZ under dichotomous utilities; In Appendix~\ref{asec.rcp} we show that this solution is captured by an RCP.

In Appendix~\ref{sec.DSIC-ADNB} we show that our mechanism for 1DLAD, given in Section \ref{sec.1DLAD} is strategyproof, provided the disagreement utilities are public knowledge and are therefore truthfully reported. 

\subsection{Related Results}\label{sec.related}

We start by stating the properties of mechanisms for one-sided matching markets listed in the Introduction. Random Priority \cite{Moulin2018fair} is strategyproof though not efficient or envy-free; Probabilistic Serial \cite{Bogomolnaia-PS} is efficient and envy-free but not strategyproof; and Top Trading Cycles \cite{Shapley1974-TTC} is efficient, strategyproof and core-stable. 

Recently, \cite{VY-HZ} undertook a comprehensive study of the computational complexity of the HZ scheme. They gave a combinatorial polynomial time algorithm for dichotomous utilities and an example which has only irrational equilibria; as a consequence, this problem is not in PPAD. They showed that the problem of computing an exact HZ equilibrium is in the class FIXP and the problem of computing an approximate equilibrium is in PPAD. They left open the problem of giving the corresponding  hardness results. 

The study of the dichotomous case of matching markets was initiated by Bogomolnaia and Moulin \cite{Bogomolnaia2004random}. They studied a two-sided matching market and they called it an ``important special case of the bilateral matching problem.'' Using the Gallai-Edmonds decomposition of a bipartite graph, they gave a mechanism that is Pareto optimal and group strategyproof. They also gave a number of applications of their setting, some of which are natural applications of one-sided markets as well, e.g., housemates distributing rooms, having different features, in a house. As in the HZ scheme, their mechanism also outputs a doubly-stochastic matrix whose entries represent probability shares of allocations. However, they give another interesting interpretation of this matrix. They say, ``Time sharing is the simplest way to deal fairly with indivisibilities of matching markets: think of a set of workers sharing their time among a set of employers.'' Roth, S\"onmez and \"Unver \cite{Roth2005-JET} extended these results to general graph matching under dichotomous utilities; this setting is applicable to the kidney exchange marketplace. 

An interesting recent paper \cite{Abebe-MM-Truthful} defines the notion of a random partial improvement mechanism for a one-sided matching market. This mechanism truthfully elicits the cardinal preferences of the agents and outputs a distribution over matchings that approximates every agent's utility in the Nash bargaining solution.

Several researchers have proposed Hylland-Zeckhauser-type mechanisms for a number of applications, for instance \cite{Budish2011combinatorial, He2018pseudo, Le2017competitive, Mclennan2018efficient}. The basic scheme has also been generalized in several different directions, including two-sided matching markets, adding quantitative constraints, and to the setting in which agents have initial endowments of goods instead of money, see  \cite{Echenique2019constrained, Echenique2019fairness}.

\section{The Hylland-Zeckhauser Mechanism} 
The Hylland-Zeckhauser (HZ) mechanism can be viewed as a marriage between a fractional perfect matching and a linear Fisher market, which consists of a set $A$ of agents and a set $G$ of goods. Each agent $i$ comes to the market with a budget $b_i$ and has utilities $u_{ij} \geq 0$ for each good $j$. In the case of linear utilities, agent $i$'s utility from allocation $(x_{ij})_{j\in G}$ is $\sum_j u_{ij}x_{ij}$. 

    A \emph{Fisher equilibrium} is a pair $(x, p)$ consisting of an \emph{allocation} $(x_{ij})_{i \in A, j \in G}$ and \emph{prices} $(p_j)_{j \in G}$ such that each agent gets a utility maximizing (optimal) bundle subject to budget constraints and market clears; see more details in Appendix~\ref{sec:fisher}. Fisher equilibria satisfy various nice properties, including equal-type envy-freeness, Pareto optimality, and approximate incentive compatibility in large markets.

\begin{definition}[Envy-freeness, Pareto optimality, and incentive compatibility]
    An allocation is \emph{envy-free} if for any two agents $i, i' \in A$, agent $i$ weakly prefers their allocation than those that $i'$ gets, i.e., $\sum_{j \in G}{u_{ij} x_{ij}} \geq \sum_{j \in G}{u_{ij} x_{i'j}}$.
    It is \emph{equal-type envy-free} if the above holds for any two agents with identical budgets.

    An allocation $x$ \emph{weakly dominates} another allocation $x'$ if no agent prefers $x'$ to $x$.
    It \emph{strongly} dominates $x'$ if it weakly dominates it and some agent prefers $x$ to $x'$.
    An allocation $x$ is \emph{Pareto efficient} or \emph{Pareto optimal} if there is no other allocation $x'$ which strongly dominates it.

    A mechanism is \emph{incentive compatible} if no agent can improve the total utility they accrue by misreporting their preferences/utilities to the mechanism.
\end{definition}

\begin{definition}
    A \emph{one-sided matching market} consists of a set $A$ of \emph{agents} and a set $G$ of \emph{goods}. Each agent has preferences over goods, expressed either using cardinal or ordinal utility functions. An {\em allocation} is a perfect matching of agents to goods. The goal of the market is to find an allocation so that the underlying mechanism has some desirable game-theoretic properties.
\end{definition}

The HZ mechanism uses cardinal utility functions, in which each good is rendered divisible by viewing it as one unit of {\em probability shares}. An HZ equilibrium is defined as follows.

\begin{definition}
    A \emph{Hylland-Zeckhauser (HZ) equilibrium} is a pair $(x, p)$ consisting of an \emph{allocation} $(x_{ij})_{i \in A, j \in G}$ and \emph{prices} $(p_j)_{j \in G}$ with the following properties.
    \begin{enumerate}
    \setlength\itemsep{0em}
        \item $x$ is a fractional perfect matching, i.e., $\sum_{j \in G}{x_{ij}} = 1$ for all $i$ and $\sum_{i \in A}{x_{ij}} = 1$ for all $j$.
        \item Each agent $i$ spends at most their budget, i.e., $\sum_{j \in G}{p_j x_{ij}} \leq b_i$ (usually $b_i = 1$).
        \item Each agent $i$ gets an \emph{optimal bundle}, which is defined to be a cheapest utility maximizing bundle, i.e.,  $\sum_{j \in G}{u_{ij} x_{ij}} = \max \left\{\sum_{j \in G}{u_{ij} y_j} \;\middle|\; \sum_j y_j = 1;\ \sum_{j \in G}{p_j y_j } \leq b_i\right\}$ and $\sum_{j \in G}{p_j x_{ij}} = \min \left\{\sum_{j \in G}{p_j y_j} \;\middle|\; \sum_j y_j = 1; \sum_{j \in G}{u_{ij} y_j} \geq \sum_{j \in G}{u_{ij} x_{ij}}\right\}$.
    \end{enumerate}
\end{definition}

Like Fisher equilibria, HZ equilibria are Pareto optimal, envy-free (assuming unit budgets), and approximately incentive compatible in large markets \cite{He2018pseudo}.\footnote{
Pareto optimality for HZ equilibria requires that each agent receives a {\em cheapest} utility maximizing bundle. If this condition is dropped, we get counter-examples to Pareto optimality: Consider an instance with two agents $a_1$ and $a_2$, and two goods $g_1$ and $g_2$ with $u_{11} = u_{21} = u_{22} = 1; u_{12} = 0$.
    The prices $(2, 0)$ together with the allocation $x_{11} = x_{12} = x_{21} = x_{22} = 0.5$ are optimal bundles, though not cheapest. The utilities in this equilibrium are $0.5$ for agent $a_1$ and $1$ for agent $a_2$. However, there is another HZ equilibrium with prices $(1,p)$, for any $p\in[0, 1]$ with utility 1 for both agents.
}
The allocation $x$ found by the HZ mechanism is a fractional perfect matching or a doubly-stochastic matrix. In order to get an integral perfect matching from $x$, a lottery can be carried out using the Theorem of Birkhoff \cite{Birkhoff1946tres} and von Neumann \cite{von1953certain}. It states that any doubly-stochastic matrix can be written as a convex combination of integral perfect matchings; moreover, this decomposition can be found efficiently. Picking a perfect matching according to the discrete probability distribution determined by this convex combination yields the resulting allocation in the HZ mechanism. 

\section{The $\epsilon$-Approximate ADHZ Model}\label{sec.ADHZ}
In this paper we are interested in an exchange version of the HZ mechanism. Before defining it, we introduce the Arrow-Debreu (exchange) market under linear utility functions, which consists of
   a set $A$ of \emph{agents} and a set $G$ of \emph{goods}.
    Each agent $i$ comes to the market with an \emph{endowment} $e_{ij} \geq 0$ of each good $j$ and also has a utility $u_{ij} \geq 0$.
    Each good $j$ must be fully owned by the agents, i.e., $\sum_{i \in A}{e_{ij}} = 1$ for all $j \in G$.

\begin{definition}
    An \emph{Arrow-Debreu (AD) equilibrium} for a given AD market is a pair $(x, p)$ consisting of an \emph{allocation} $(x_{ij})_{i \in A, j \in G}$ and \emph{prices} $(p_j)_{j \in G}$ with the following properties.
    \begin{enumerate}
    \setlength\itemsep{-0.2em}
        \item Each agent spends at most the budget earned from the endowment, i.e., $\sum_{j}{p_j x_{ij}} \leq b_i \coloneqq \sum_{j}{p_j e_{ij}}$.
        \item Each agent $i$ gets an \emph{optimal bundle}, i.e., utility maximizing bundle at $p$. Formally:
                $\sum_{j \in G}{u_{ij} x_{ij}} = \max \left\{\sum_{j \in G}{u_{ij} y_j} \;\middle|\; y \in \bbR_{\geq 0}^G, \sum_{j \in G}{p_j y_j } \leq b_i\right\}.$
        \item The market clears, i.e., each good with positive price is fully allocated to the agents.
    \end{enumerate}
\end{definition}

The AD model generalizes Fisher model in the sense that any Fisher market can be easily transformed into an AD market by giving each agent a fixed proportion of every good. Clearly, AD equilibria satisfy the condition of individual rationality, defined below, since every agent could always buy back their endowment.

\begin{definition}
    An allocation in an AD market is \emph{individually rational} if for every agent $i$ we have $\sum_{j}{u_{ij} x_{ij}} \geq \sum_{j}{u_{ij} e_{ij}}$, i.e., no agent loses utility by participating in the market.
\end{definition}

However, individual rationality fundamentally clashes with envy-freeness.
Consider a market consisting of two agents each owning a distinct good.
Assume that both agents prefer the good of agent 2 over the good of agent 1, then in any allocation either agent 1 envies agent 2 or agent 2's individual rationality is violated.
For this reason we primarily consider a version of equal-type envy-freeness in exchange markets, which demands envy-freeness only for agents with the same initial endowment.

AD equilibria do not always exist. However, there is a simple necessary and sufficient condition for their existence based on \emph{strong connectivity of demand graph}, due to Gale \cite{Gale76}. 
An RCP for this problem was given by Devanur, Garg and V\'egh \cite{DevanurGV16}.

We now turn to the extension of the HZ mechanism to exchange markets.
    In the \emph{ADHZ market}, we have a set $A$ of \emph{agents} and a set $G$ of \emph{goods} with $|A| = |G| = n$.
    Each agent $i$ comes with an \emph{endowment} $e_{ij} \geq 0$ of each good $j$ and utilities $u_{ij} \geq 0$.
    The endowment vector $e$ is a fractional perfect matching.

\begin{definition}
    An \emph{ADHZ equilibrium} for a given ADHZ market is a pair $(x, p)$ consisting of an \emph{allocation} $(x_{ij})_{i \in A, j \in G}$ and \emph{prices} $(p_j)_{j \in G}$ with the following properties.
    \begin{enumerate}
    \setlength\itemsep{-0.2em}
        \item $x$ is a fractional perfect matching, i.e., $\sum_{j \in G}{x_{ij}} = 1$ for all $i$ and $\sum_{i \in A}{x_{ij}} = 1$ for all $j$.
        \item Each agent spends at most the budget earned from the endowment, i.e., $\sum_{j}{p_j x_{ij}} \leq b_i \coloneqq \sum_{j}{p_j e_{ij}}$.
        \item Each agent $i$ gets an \emph{optimal bundle}, which is defined to be a cheapest utility maximizing bundle, i.e., 
                $\sum_{j \in G}{u_{ij} x_{ij}} = \max \left\{\sum_{j \in G}{u_{ij} y_j} \;\middle|\; \sum_j y_j = 1;\ \sum_{j \in G}{p_j y_j } \leq b_i\right\}$  and $\sum_{j \in G}{p_j x_{ij}} = \min \left\{\sum_{j \in G}{p_j y_j} \;\middle|\; \sum_j y_j = 1;\  \sum_{j \in G}{u_{ij} y_j} \geq \sum_{j \in G}{u_{ij} x_{ij}}\right\}$.
    \end{enumerate}
\end{definition}

\begin{theorem}
    ADHZ equilibria are Pareto optimal, individually rational, and equal-type envy-free.
\end{theorem}

\begin{proof}
    Pareto optimality follows from the fact that any ADHZ equilibrium is an HZ equilibrium with certain budgets $b$.
    Since any HZ equilibrium is Pareto optimal, we get the same for ADHZ. 

    Note that the budget of any agent is always enough to buy back their initial endowment.
    Since they get an optimal bundle, they must get something which they value at least as high as their initial endowment.
    Thus individual rationality is guaranteed.

    If two agents, say 1 and 2, have the same endowment, then their budget will be the same and so agent 1 will never value the 2's bundle higher than their own.
    Thus ADHZ equilibria are equal-type envy-free.
\end{proof}

In addition, ADHZ equilibria also satisfy the following notion of core-stability.

\begin{definition}
    An allocation $x$ in an ADHZ market is \emph{weakly core-stable} if for any subsets $A' \subseteq A$ and $G' \subseteq G$, there does not exist an allocation $x' \in \mathbb{R}^{A' \times G'}_{\geq 0}$ such that
    \begin{itemize}
    \setlength\itemsep{-0.3em}
        \item $x'$ allocates at most one unit of goods to every agent in $A'$,
        \item every good $j \in G'$ is allocated at most to the extent of the endowments of the agents in $A'$, i.e.\ $\sum_{i \in A'}{x'_{ij}} \leq \sum_{i \in A'}{e_{ij}}$, and
        \item every agent in $A'$ receives strictly better utility in $x'$ than in $x$.
    \end{itemize}
\end{definition}

\begin{theorem}\label{thm:cs}
    ADHZ equilibria are weakly core-stable.
\end{theorem}

\begin{proof}
    Let $(x, p)$ be some ADHZ equilibrium.
    For the sake of a contradiction, assume that there are $A' \subseteq A$, $G' \subseteq G$ and $x' \in \mathbb{R}^{A' \times G'}_{\geq 0}$ as excluded by the definition of weak core-stability.
    Now consider the total money spent \enquote{along allocation $x'$}, i.e., the quantity $\sum_{i \in A'} \sum_{j \in G'} {p_{j} x'_{ij}}$.

    On the one hand we know that only the endowment of the agents in $A'$ is allocated by $x'$.
    Thus
    \[
        \sum_{i \in A'} \sum_{j \in G'} {p_{j} x'_{ij}} \leq \sum_{i \in A'} \sum_{j \in G'} {p_{j} e_{ij}}.
    \]

    On the other hand, every agent $i$ receives strictly better utility from $x'$ than from $x$.
    But since agents buy optimal bundles in $(x, p)$, this implies that the bundles in $x'$ must be worth more than their budget, i.e.,
    \[
        \sum_{j \in G'} {p_{j} x'_{ij}} > \sum_{j \in G}{p_j e_{ij}} \geq \sum_{j \in G'}{p_j e_{ij}}.
    \]
    Summing this inequality over all $i \in A'$ yields a contradiction to the previous inequality.
\end{proof}

Like in the case of HZ, equilibrium prices in ADHZ are invariant under the operation of \emph{scaling} the difference of prices from $1$, as shown in the following lemma.  

\begin{lemma}\label{lem:scale}
Suppose $p$ be an equilibrium price vector. For any $r>0$, let $p'$ be such that $p_j' - 1 = r(p_j - 1)$ for all $j\in G$. Then $p'$ is also an equilibrium price vector. 
\end{lemma}

\begin{proof}
Let $x$ be an equilibrium allocation at prices $p$. For any agent $i$, we have $\sum_{j\in G} x_{ij}p_j \le \sum_{j\in G} e_{ij}p_j$. We show that the pair $(x, p')$ is also an equilibrium. 

Since $(x,p)$ is an equilibrium, we have
{\small \[\forall i\in A:\ \ \sum_{j\in G} u_{ij}x_{ij} = \max\left\{\sum_{j\in G} u_{ij}y_j\ |\ y\in\mathbb{R}^G_{\ge 0},\ \sum_{j\in G} y_j = 1,\ \sum_{j\in G} y_jp_j \le \sum_{j\in G} e_{ij}p_j\right\}.\] }
Replacing $p_j$ by $(p_j' - 1)/r + 1$ for all $j\in G$, we get:
{\small \[\forall i\in A:\ \ \sum_{j\in G} u_{ij}x_{ij} = \max\left\{\sum_{j\in G} u_{ij}y_j\ |\ y\in\mathbb{R}^G_{\ge 0},\ \sum_{j\in G} y_j = 1,\ \sum_{j\in G} y_j \left(\frac{p_j'-1}{r} + 1 \right) \le \sum_{j\in G} e_{ij} \left(\frac{p_j'-1}{r}+1 \right) \right\}.\] }
Simplifying the above using $\sum_{j\in G} e_{ij} = 1$ and $\sum_{j\in G} y_j = 1$ for all $i\in A$, we get:
{\small \[\forall i\in A:\ \ \sum_{j\in G} u_{ij}x_{ij} = \max\left\{\sum_{j\in G} u_{ij}y_j\ |\ y\in\mathbb{R}^G_{\ge 0},\ \sum_{j\in G} y_j = 1,\ \sum_{j\in G} y_jp_j' \le \sum_{j\in G} e_{ij}p_j'\right\}.\] }
The above implies that $x$ gives each agent an optimal bundle at prices $p'$. This, together with the fact that $x$ is a fractional perfect matching, shows that $(x,p')$ is also an equilibrium. 
\end{proof}

Unlike HZ, which always admits an equilibrium, ADHZ has instances which do not admit an equilibrium, as observed by Hylland and Zeckhauser \cite{HZ79}. 
Below we give a counterexample in which the demand graph is strongly connected and utilities are dichotomous.

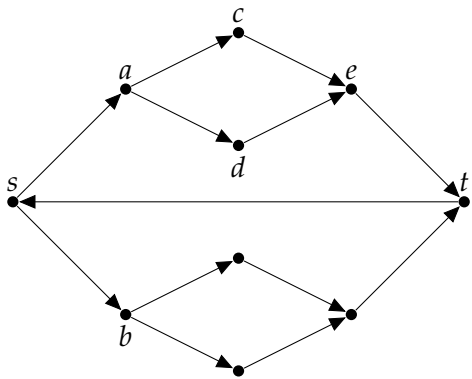
\begin{figure}[t]
\vskip -1cm
    \begin{center}
        \begin{tikzpicture}[scale=1.5]
                \node[fill, circle, inner sep=1.5pt] (v1) at (0, 0) {};
                \node[fill, circle, inner sep=1.5pt] (v2) at (4, 0) {};

                \node[fill, circle, inner sep=1.5pt] (v3) at (1, 1) {};
                \node[fill, circle, inner sep=1.5pt] (v4) at (1, -1) {};

                \node[fill, circle, inner sep=1.5pt] (v5) at (2, 1.5) {};
                \node[fill, circle, inner sep=1.5pt] (v6) at (2, 0.5) {};
                \node[fill, circle, inner sep=1.5pt] (v7) at (2, -0.5) {};
                \node[fill, circle, inner sep=1.5pt] (v8) at (2, -1.5) {};

                \node[fill, circle, inner sep=1.5pt] (v9) at (3, 1) {};
                \node[fill, circle, inner sep=1.5pt] (v10) at (3, -1) {};

                \draw[arrows={-triangle 45}] (v1) -- (v3);
                \draw[arrows={-triangle 45}] (v1) -- (v4);
                \draw[arrows={-triangle 45}] (v3) -- (v5);
                \draw[arrows={-triangle 45}] (v3) -- (v6);
                \draw[arrows={-triangle 45}] (v4) -- (v7);
                \draw[arrows={-triangle 45}] (v4) -- (v8);
                \draw[arrows={-triangle 45}] (v5) -- (v9);
                \draw[arrows={-triangle 45}] (v6) -- (v9);
                \draw[arrows={-triangle 45}] (v7) -- (v10);
                \draw[arrows={-triangle 45}] (v8) -- (v10);
                \draw[arrows={-triangle 45}] (v9) -- (v2);
                \draw[arrows={-triangle 45}] (v10) -- (v2);
                \draw[arrows={-triangle 45}] (v2) -- (v1);

                \node[above] at (v1) {$s$};
                \node[above] at (v2) {$t$};

                \node[above] at (v3) {$a$};
                \node[below] at (v4) {$b$};
                \node[above] at (v5) {$c$};
                \node[below] at (v6) {$d$};
                \node[above] at (v9) {$e$};
        \end{tikzpicture}
    \end{center}
\vskip -0.5cm
    \caption{The demand graph of an ADHZ market with dichotomous utilities and no equilibrium. Each node represents an agent as well as the good possessed by this agent in her initial endowment. An arrow from $i$ to $j$ represents $u_{ij} = 1$; the rest of the  edges have utility 0. \label{fig:no_equilibrium}}
\end{figure}

\begin{proposition}\label{prop:counter_example}
    The ADHZ market with dichotomous utilities in Figure~\ref{fig:no_equilibrium} does not admit an equilibrium.
\end{proposition}
\begin{proof}    
Assume there is an equilibrium $(x,p)$ in this market. Further, using Lemma~\ref{lem:scale}, we can assume that the minimum price is zero at $p$. This implies that no agent will buy a zero utility good at a positive price. 

Each agent buys a total of one unit of goods and $s$ is the only agent having positive utility for goods $a$ and $b$. Therefore, at least one of these goods is not fully sold to $s$ and must be sold to an agent deriving zero utility from it. Therefore this good must have zero price.  Without loss of generality, assume $p_a = 0$. Since $a$ has no budget and $c$ and $d$ are desired only by $a$, $p_c = p_d = 0$, otherwise $c$ and $d$ cannot be sold.
    For the same reason, $p_e = 0$.
     Now observe that both agents $c$ and $d$ have a utility 1 edge to a good of price zero, namely $e$. Therefore, the optimal bundle of both $c$ and $d$ is $e$.
    But then $e$ would have to be matched twice which is a contradiction.
\end{proof}

Even if ADHZ equilibria \emph{do} exist, computing them is at least as hard as computing HZ equilibria. This follows from the following reduction.

\begin{proposition}
\label{prop.reduction}
    Consider an HZ market with unit budgets.
    Define an ADHZ market by giving every agent as endowment an equal amount of every good.
    Then every HZ equilibrium in which the prices sum up to $n$ is an ADHZ equilibrium and every ADHZ equilibrium yields an HZ equilibrium by rescaling all prices by $n / \sum_{j \in G}{p_j}$.
\end{proposition}

\cite{VY-HZ} gave an instance of HZ with four agents and four goods which has one equilibrium in which all agents fully spend their budgets, and allocations and prices are irrational. Since this example satisfies the conditions of Proposition~\ref{prop.reduction}, we get that the modification of the example of \cite{VY-HZ}, as stated in the Proposition, is an instance for ADHZ having only irrational equilibria.

\subsection{Existence and Properties of \texorpdfstring{$\epsilon$}{epsilon}-Approximate ADHZ Equilibria}
\label{sec.epsilon}

Since ADHZ equilibria do not always exist we study the following approximate equilibrium notion instead.

\begin{definition}
    An \emph{$\epsilon$-approximate ADHZ equilibrium} is an HZ equilibrium $(x, p)$ for a budget vector $b$ with
    \[
        (1 - \epsilon) \sum_{j \in G}{p_j e_{ij}} \leq b_i \leq \epsilon + \sum_{j \in G}{p_j e_{ij}} \ \ \ \text{ for all } i\in A\enspace . 
    \]
    We also require that if two agents have the same endowment, then their budget should also be the same.\footnote{
    The additive error term on the upper bound is necessary as otherwise the counterexample 
    from Proposition~\ref{prop:counter_example} still works.
    On the other hand, the multiplicative lower bound is useful to get approximate individual rationality.
    However, one can always find approximate equilibria in which the sum of prices is bounded by $n$, so this implies
$      \sum_{j \in G}{p_j e_{ij}} - \epsilon' \leq b_i \leq \sum_{j \in G}{p_j e_{ij}} + \epsilon'
    \ \ \ \  \text{for}\ \ \epsilon' \coloneqq n \epsilon.$
    Further, by using the instance in Figure~\ref{fig:no_equilibrium} as a gadget, it is easy to construct very restricted ADHZ markets which require a constant fraction of agents to get money over their ADHZ budget in order for an equilibrium to exist.
    Thus, it is not possible to simply inject some money into one vertex of every strongly connected component in order to obtain an equilibrium.
}

\end{definition}

In our notion of approximate equilibrium, we do not relax the fractional perfect matching constraints or the optimum bundle condition.
We only allow the budgets of agents to be slightly different from the money they would normally obtain in an ADHZ market. Hence the step of randomly rounding the equilibrium allocation to an integral perfect matching is the same as in the HZ scheme.

\begin{theorem}
    Any $\epsilon$-approximate ADHZ equilibrium is Pareto optimal, $\epsilon$-approximately individually rational, equal-type envy-free.
\end{theorem}

\begin{proof}
    Pareto optimality follows just as for the non-approximate ADHZ setting from the fact that an $\epsilon$-approximate ADHZ equilibrium is first and foremost an HZ equilibrium.
    For approximate individual rationality note that every agent gets a budget of at least $(1 - \epsilon)$ times the cost of their endowment.
    Hence their utility can decrease by at most a factor of $(1 - \epsilon)$.
    Equal-type envy-freeness follows immediately from the condition that agents with the same endowment have the same budget.
\end{proof}

One can also define a suitably $\epsilon$-approximate notion of weak core-stability, where instead of demanding that every agent strictly improves in the seceding coalition, we instead require that every agent improves by a factor of more than $\frac{1}{1 - \epsilon}$.

\begin{theorem}
    Any $\epsilon$-approximate ADHZ equilibrium is $\epsilon$-approximately weak-core stable.
\end{theorem}

\begin{proof}
    Let $(x, p)$ be an $\epsilon$-approximate ADHZ equilibrium for some budget vector $b$.
    Then in order for some other allocation $x'$ to improve agent $i$'s utility by a factor of more than $\frac{1}{1 - \epsilon}$, $i$ must spend more than $\frac{b_i}{1 - \epsilon}$.
    But note that
        $\frac{b_i}{1 - \epsilon} \geq \sum_{j \in G}{p_j e_{ij}}.$
    From here the proof is identical to that of Theorem~\ref{thm:cs}. 
\end{proof}

While approximate equilibrium notions are more amenable to computation, they generally do not lend themselves well to existence proofs.
However, our notion of $\epsilon$-approximate ADHZ equilibrium is a slight relaxation of the notion of an $\alpha$-slack equilibrium introduced in \cite{Echenique2019constrained}.

\begin{definition}
    An \emph{$\alpha$-slack ADHZ equilibrium} for $\alpha \in (0, 1]$ is an HZ equilibrium $(x, p)$ for a budget vector $b$ in which $b_i = \alpha + (1 - \alpha) \sum_{j \in G}{p_j e_{ij}}$ for all $i \in A$.
\end{definition}

\begin{theorem}[Theorem 2 in \cite{Echenique2019constrained}]
    In any ADHZ market, $\alpha$-slack equilibria always exist if $\alpha > 0$.
\end{theorem}

Note that any $\alpha$-slack equilibrium is automatically also an $\alpha$-approximate equilibrium.
Thus we get: 

\begin{theorem}
    In any ADHZ market, $\epsilon$-approximate equilibria always exist if $\epsilon > 0$.
\end{theorem}

\section{Algorithm for $\epsilon$-approximate ADHZ under Dichotomous Utilities}
\label{sec.alg}

Before we can tackle the ADHZ setting, let us first give an algorithm that can compute HZ equilibria with non-uniform budgets.
This is an extension of the algorithm presented in \cite{VY-HZ}.
In the following, fix some HZ market consisting of $n$ agents and goods with $u_{ij} \in \{0, 1\}$ for all $i \in A$ and $j \in G$.
If $u_{ij} = 1$, we will say that $i$ \emph{likes} $j$ (and \emph{dislikes} otherwise).
We assume that every agent likes at least one good.
\footnote{
    Any HZ equilibrium $(x, p)$ for the utilities $u_{ij}$ is also an equilibrium for $\tilde{u}_{ij}$ where
    $\tilde{u}_{ij} = a_i \text{ if } u_{ij} = 0$ and $b_i$ if $u_{ij} = 1$ for 
    for all agents $i$, goods $j$, and arbitrary $0 \leq a_i < b_i$ for every agent.
    This is because
   $     \sum_{j \in G}{\tilde{u}_{ij} x_{ij}} = a_i + (b_i - a_i) \sum_{j \in G}{u_{ij} x_{ij}}$
    since $x$ is a fractional perfect matching.
    Hence utility function $\tilde{u}$ is an affine transformation of utility function \ $u$; the former is called a {\em bi-valued utility function}.
}\label{footnote1}

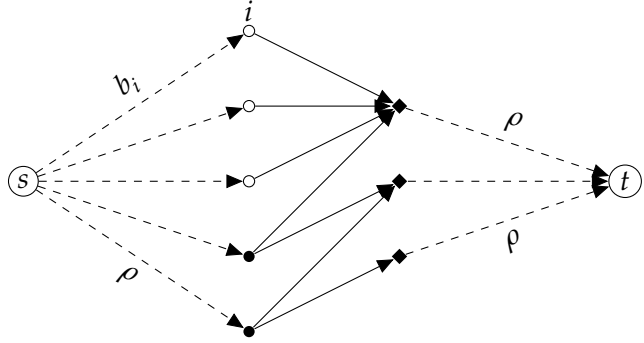
\begin{figure}[hbt]
\vskip -0.5cm
    \begin{center}
        \begin{tikzpicture}[scale=1.0]
            \node[draw, circle, inner sep=1.5pt] (s) at (-1, 0) {$s$};

            \node[draw, circle, inner sep=1.5pt] (b1) at (2, 2) {};
            \node[draw, circle, inner sep=1.5pt] (b2) at (2, 1) {};
            \node[draw, circle, inner sep=1.5pt] (b3) at (2, 0) {};
            \node[fill, circle, inner sep=1.5pt] (s1) at (2, -1) {};
            \node[fill, circle, inner sep=1.5pt] (s2) at (2, -2) {};

            \node[fill, diamond, inner sep=1.5pt] (c1) at (4, -1) {};
            \node[fill, diamond, inner sep=1.5pt] (c2) at (4, 0) {};
            \node[fill, diamond, inner sep=1.5pt] (c3) at (4, 1) {};

            \node[draw, circle, inner sep=1.5pt] (t) at (7, 0) {$t$};

            \draw[dashed, arrows={-triangle 45}] (s) edge node[midway, sloped, above] {$b_i$} (b1);
            \draw[dashed, arrows={-triangle 45}] (s) -- (b2);
            \draw[dashed, arrows={-triangle 45}] (s) -- (b3);
            \draw[dashed, arrows={-triangle 45}] (s) -- (s1);
            \draw[dashed, arrows={-triangle 45}] (s) edge node[midway, sloped, below] {$\rho$} (s2);

            \draw[dashed, arrows={-triangle 45}] (c1) edge node[midway, sloped, below] {$\rho$} (t);
            \draw[dashed, arrows={-triangle 45}] (c2) -- (t);
            \draw[dashed, arrows={-triangle 45}] (c3) edge node[midway, sloped, above] {$\rho$} (t);

            \draw[arrows={-triangle 45}] (b1) -- (c3);
            \draw[arrows={-triangle 45}] (b2) -- (c3);
            \draw[arrows={-triangle 45}] (b3) -- (c3);
            \draw[arrows={-triangle 45}] (s1) -- (c3);
            \draw[arrows={-triangle 45}] (s1) -- (c2);
            \draw[arrows={-triangle 45}] (s2) -- (c2);
            \draw[arrows={-triangle 45}] (s2) -- (c1);

            \node[above] at (b1) {$i$};
        \end{tikzpicture}
    \end{center}
\vskip -0.5cm
    \caption{Shown is the flow network which corresponds to finding an equilibrium allocation in price class $\rho$. Filled circles represent agents in $A(\rho)$ with $b_i < \rho$, empty circles are agents in $A(\rho)$ with $b_i \geq \rho$, and diamond vertices are goods in $G(\rho)$. The contiguous edges represent all utility 1 edges and have infinity capacity (utility 0 edges are not part of the network). Dashed edges to empty circle vertices $i$ have capacity $b_i$ whereas the other dashed edges have capacity $\rho$.\label{fig:flow_network}}
\end{figure}

\begin{lemma}\label{lem:flow_implies_allocation}
    Let $(p_j)_{j \in G}$ be non-negative prices.
    For any $\rho \geq 0$, let $G(\rho)$ be the goods which are sold at price $\rho$ and let $A(\rho)$ be those agents for which the cheapest price of any liked good is $\rho$.
    Assume that
    \begin{itemize}
    \setlength\itemsep{-0.2em}
        \item there is a matching in the utility 1 edges on $A(0) \cup G(0)$ which covers all agents in $A(0)$ and
        \item if $\rho > 0$ is equal to the price of some good, then the flow network shown in Figure~\ref{fig:flow_network} has a maximum flow of size $\rho |G(\rho)|$.
    \end{itemize}
    Then we can find a fractional perfect matching $x$ which makes $(x, p)$ an HZ equilibrium in polynomial time.
\end{lemma}

\begin{proof}
    Allocate every agent in $A(0)$ to some good in $G(0)$ according to the matching which exists by assumption.
    Let $\rho > 0$, be the price of some good.
    Then we compute the maximum flow $f^{(\rho)}$ in the flow network from Figure~\ref{fig:flow_network} and allocate $x_{ij} = f^{(\rho)}_{i, j} / \rho$ for all $i \in A(\rho)$ and $j \in G(\rho)$.
    Lastly, extend $x$ to a fractional perfect matching by matching the remaining capacity of the agents to the remaining capacity of goods in $G(0)$.

    Clearly, no agent exceeds their budget.
    To see that this yields an HZ equilibrium, note that every agent only spends money on cheapest liked goods and if they do not get allocated entirely to liked goods, then they additionally spend all of their budget.
    This ensures that every agent gets an optimum bundle.
\end{proof}

\begin{theorem}\label{thm:hz_compute}
    For any rational budget vector $b$, we can compute an HZ equilibrium in polynomial time.
\end{theorem}

\begin{proof}
    We start in the same way is the algorithm in \cite{VY-HZ}: by computing a minimum vertex cover in the graph of utility 1 edges, we partition $A = A_1 \cup A_2$ and $G_1 \cup G_2$ such that
    \begin{itemize}
    \setlength\itemsep{-0.2em}
        \item every agent in $A_2$ can be matched to a distinct liked good in $G_2$,
        \item every agent in $A_1$ only has liked goods in $G_1$, and
        \item for every $S \subseteq G_2$ we have $|N^{-}(S)| \geq |S|$ where $N^{-}(S)|$ are the agents that have a liked good in $S$.
    \end{itemize}

    Set $p_j = 0$ for all $j \in G_2$ and $p_j = \min_{i \in A_1}{b_i}$ for all $j \in G_1$.
    Now we run a DPSV-like \cite{DPSV08} algorithm on $A_1 \cup G_1$ to raise prices until certain sets of goods become tight.

    For each $i \in A$, let $\beta_i$ be its \emph{effective} budget at current prices $p$, that is the minimum of its actual budget $b_i$ and the price of its cheapest liked good.
    The algorithm will now raise all prices $p$ at the same rate until there is a set $S \subseteq G_1$ which goes \emph{tight} in the sense that
$        \sum_{i \in \Gamma(S)}{\beta_i} = \sum_{j \in S}{p_j}$ 
    where $\Gamma$ is the collection of agents which have a cheapest liked good in $S$.
    At this point, we freeze the prices of the goods in $S$.
    If all prices have been frozen we are done.
    Otherwise, we continue raising all unfrozen prices of goods in $G_1$.

    It is easy to see that if the prices keep rising, eventually each agents' effective budget will be their real budget and so a set must become tight at some point.
    We will not go into detail here but it is possible to find the next set which will go tight in polynomial time similar as in DPSV.
    Finally, since we never unfreeze prices, there will be at most $n$ iterations of the algorithm and hence it runs in polynomial time overall.

    We observe that as in the proof of the DPSV algorithm, for any $S \subseteq G_1$, we have that $\sum_{i \in \Gamma(S)}{\beta_i} \geq \sum_{j \in A}{p_j}$ and $\sum_{i \in A_1}{\beta_i} = \sum_{j \in G_1}{p_j}$.
    It is then easy to show that this implies that for any price $\rho$ above 0, the corresponding flow network from Figure~\ref{fig:flow_network} supports a flow of value $\rho |G(\rho)|$ by the max-flow min-cut theorem.
    Thus we can apply Lemma~\ref{lem:flow_implies_allocation} to get an equilibrium allocation.
\end{proof}

\begin{lemma}\label{lem:next_hz_compute}
    Let $b$ and $b'$ be two budget vectors with $0 \leq b \leq b'$.
    Assume we are given an HZ equilibrium $(x, p)$ for the budgets $b$.
    Then we can compute in polynomial time a new HZ equilibrium $(x', p')$ with $p \leq p'$ for the budgets $b'$.
\end{lemma}

\begin{proof}
    We will simply run the same algorithm as in the proof of Theorem~\ref{thm:hz_compute}, except that this time we start with the prices $p$.
    More precisely, we increase the lowest non-zero price until a set goes tight or it becomes equal to the next higher price, then repeat this process until we once again get $\sum_{i \in \Gamma(S)}{\beta_i} \geq \sum_{j \in A}{p_j}$ and $\sum_{i \in \Gamma(G_1)}{\beta_i} = \sum_{j \in G_1}{p_j}$ where $G_1$ is now defined as the set of goods with positive prices in $(x, p)$.
    As in the proof of Theorem~\ref{thm:hz_compute}, this will freeze all prices in polynomial time at which point we can use a max-flow min-cut argument to construct the new equilibrium allocation $x'$ in polynomial time.
\end{proof}

Let us now return to the approximate ADHZ setting.
Instead of budgets, fix now some fractional perfect matching of endowments $(e_{ij})_{i \in A, j \in G}$.

\begin{theorem}
    An $\epsilon$-approximate ADHZ equilibrium for rational $\epsilon \in (0, 1)$, can be computed in time polynomial in $\frac{1}{\epsilon}$ and $n$, i.e.\ by a fully polynomial time approximation scheme.
\end{theorem}

\begin{proof}
    We will iteratively apply Lemma~\ref{lem:next_hz_compute}.
    Start by setting $b^{(1)}_i \coloneqq \frac{\epsilon}{2}$ for all $i \in A$ and computing an HZ equilibrium $(x^{(1)}, p^{(1)})$ according to Theorem~\ref{thm:hz_compute}.
    Beginning with $k \coloneqq 1$, we run the following algorithm.
    \begin{enumerate}
    \setlength\itemsep{-0.1em}
        \item Let $b^{(k + 1)}_i \coloneqq \frac{\epsilon}{2} + (1 - \frac{\epsilon}{2}) \sum_{j \in G}{p^{(k)}_j e_{ij}}$ for all $i \in A$.
        \item Compute a new HZ equilibrium $(x^{(k + 1)}, p^{(k + 1)})$ for budgets $b^{(k + 1)}$ according to Lemma~\ref{lem:next_hz_compute} using the old equilibrium $(x^{(k)}, p^{(k)})$ as the starting point.
            Note that since $p^{(k)} \geq p^{(k - 1)}$ we always have $b^{(k + 1)} \geq b^{(k)}$ and so this is well-defined.
        \item Set $k \coloneqq k + 1$ and go back to step 1.
    \end{enumerate}

    Note that
    \begin{align*}
        \sum_{i \in A}{b^{(k + 1)}_i} &= \frac{\epsilon}{2} n + \left(1 - \frac{\epsilon}{2}\right) \sum_{j \in G}{p^{(k)}_j} \ \ \
                                      \leq \frac{\epsilon}{2} n + \left(1 - \frac{\epsilon}{2}\right) \sum_{i \in A}{b^{(k)}_i}
    \end{align*}
    and thus
    \[
        \sum_{j \in G}{p^{(k)}_j} \leq \sum_{i \in A}{b^{(k)}_i} \leq n
    \]
    as otherwise we would get $\sum_{i \in A}{b^{(k + 1)}_i} < \sum_{i \in A}{b^{(k)}_i}$.

    Let $K$ be the first iteration such that $p^{(K)} \leq \frac{1 - \epsilon / 2}{1 - \epsilon} p^{(K - 1)}$.
    Note that
    \[
        K \leq n \log_{\frac{1 - \epsilon / 2}{1 - \epsilon}}\left(\frac{n}{\epsilon}\right) = O\left(\frac{n}{\epsilon} \log\left(\frac{n}{\epsilon}\right)\right)
    \]
    since all non-zero prices are initialized to at least $\epsilon$ but are bounded by $n$.
    Then $(x^{(K)}, p^{(K)})$ is an $\epsilon$-approximate ADHZ equilibrium with budget vector $b^{(K)}$ because for all $i \in A$ we have
    \begin{align*}
        b^{(K)}_i &= \frac{\epsilon}{2} + \left(1 - \frac{\epsilon}{2}\right) \sum_{j \in G}{p^{(K - 1)}_j e_{ij}} 
                  \ \ \ \in \left[(1 - \epsilon) \sum_{j \in G}{p^{(K)}_j e_{ij}}, \epsilon + \sum_{j \in G}{p^{(K)}_j e_{ij}}\right].
    \end{align*}

    Lastly, we note that since the number of iterations is bounded by $O(\frac{n}{\epsilon} \log(\frac{n}{\epsilon}))$ and each iteration runs in polynomial time, the total runtime is polynomial in $\frac{1}{\epsilon}$ and $n$ as claimed.
\end{proof}

\section{The Model {\em 1DLAD}}
\label{sec.1DLAD}

In this section, we define the model {\em Nash-bargaining-based one-sided matching market with initial endowments for dichotomous utilities}, which we abbreviate to {\em 1DLAD}. We start with the setting of the ADHZ scheme under dichotomous utilities and enhance it as follows: we are given a fractional perfect matching $x_E$ which specifies the initial endowments of all agents, each agent getting a total of one unit of goods. Clearly $x_E$ and the utility functions of all agents define the utility accrued by each agent from her initial endowment. We will take this to be agent $i$'s disagreement point $c_i$ and will interpret the problem as a Nash bargaining problem; see Appendix~\ref{sec.Nash} for the details on the Nash bargaining problem. 

The feasible set $\CN$ is defined as follows. Let $x$ be a fractional perfect matching over the agents $A$ and goods $G$ and let $v_x$ be an $n$-dimensional vector whose components are the utilities derived by the agents under the allocation $x$. Then $\CN$ is the set of vectors $v_x$ corresponding to all fractional perfect matchings $x$.

As in the Nash bargaining problem, we will assume that the problem is feasible, i.e., there is a fractional perfect matching, defining a redistribution of the goods, under which each agent $i$ derives strictly more utility than $c_i$. Henceforth, we consider the slightly more general problem in which the disagreement point $c$ is specified without reference to initial endowments $x_E$. In fact, there is no guarantee that $c$ comes from a valid fractional perfect matching of initial endowments. The only restriction is that the problem is feasible.

In Section~\ref{sec:rcp-1dlad}, we present a rational convex program (RCP) for {\em 1DLAD} under dichotomous utilities.  In Section \ref{sec.Variable}, we use insights obtained from the RCP to define a market model whose equilibria capture optimal solutions to RCP and hence to {\em 1DLAD}. The model is a variant of the linear Fisher market; we will call it the {\em variable-budget market model} and denote by $\CM$. In Section \ref{sec.mechanism-ADNB}, we will then present an efficient combinatorial algorithm for computing an equilibrium for $\CM$; this is also the mechanism for {\em 1DLAD}. In Appendix~\ref{sec.DSIC-ADNB} we prove that this mechanism is strategyproof under the assumption that the disagreement utilities $c_i$'s are public knowledge and therefore are truthfully reported. If they can also be tampered with, we leave the problem of obtaining a strategyproof mechanism open; see Section \ref{sec.discussion}.

\subsection{Rational Convex Program}\label{sec:rcp-1dlad}
In this section we will show that program (\ref{eq.CP-LAD-copy}) for the model {\em 1DLAD} is a RCP for the case of dichotomous utilities. This will establish the useful property the model always admits an optimal solution using rational numbers --- a pre-requisite for seeking a combinatorial, efficient algorithm. In addition, it will provide important insights into the nature of the dual variables $p_j$ and $q_i$; these will help in defining a market model whose equilibria correspond to the optimal solutions of \eqref{eq.CP-LAD-copy}. Finally, it gives a property of optimal solutions, presented in Corollary \ref{cor.concave}.
\begin{equation}\label{eq.CP-LAD-copy}
\begin{aligned}
\max &\sum_{i \in A}  {\log (v_i - c_i)} & \\ 
\text{ s.t. } & v_i= \sum_{j} {u_{ij} x_{ij}}, & \forall i \in A \\
& \sum_{j} x_{ij}\le 1, & \forall i \in A \\
& \sum_{i} x_{ij}\le 1, & \forall j \in G \\
& x_{ij}\geq 0, & \forall i \in A, \forall j \in G 
\end{aligned}
\end{equation}

The KKT conditions for program (\ref{eq.CP-LAD-copy}) are:

\begin{enumerate}
    \setlength\itemsep{-0.2em}
\item $\forall i \in A: \ q_i \geq 0$.
\item $ \forall j \in G: \ p_j \geq0$.
\item $\forall i \in A: \ q_i > 0  \  \implies \ \sum_{j} {x_{ij}} = 1$.

\item $ \forall j \in G: \ p_j > 0  \  \implies \ \sum_{i} {x_{ij}} = 1$.

\item $ \forall i \in A, \ \forall j \in G: \ p_j + q_i \geq {u_{ij} \over {v_i - c_i}}$.

\item $ \forall i \in A, \ \forall j \in G:  x_{ij} > 0 \ \implies p_j + q_i = {u_{ij} \over {v_i - c_i}}$.
\end{enumerate}

The proof of the following theorem is given in Appendix~\ref{sec.RCP-1LAD}.

\begin{theorem}
\label{thm.RCP-1}
Program (\ref{eq.CP-LAD-copy}) is a rational convex program for the model Nash-bargaining-based one-sided matching market with initial endowments, {\em 1DLAD}, under dichotomous utilities. 
\end{theorem}

\subsection{The Variable-Budget Market Model}
\label{sec.Variable}

We will define this model for arbitrary utilities, i.e., not just dichotomous utilities, even though the algorithm will be only for the latter case. In doing so, we will gain an understanding of the fairness property of the Nash bargaining solution concept. Additionally, it will open up further algorithmic questions; see Section \ref{sec.discussion}.

Let $I = (n, A, G, u, c)$ be the given instance of {\em 1DLAD}, where $u$ gives the utilities of the $n$ agents $A$ for the $n$ goods $G$ and $c$ is the disagreement point. As stated in Section \ref{sec.1DLAD}, we will assume that $I$ is feasible. In the corresponding variable-budget market $\CM =(n, A, G, u, c)$, $c_i$ specifies a {\em strict lower bound on the utility} that agent $i$ must derive from an equilibrium allocation. 

In the market $\CM$, the total utility derived by $i$ from allocation $x$ is the same as in a linear Fisher market, i.e., $\sum_j {u_{ij} x_{ij}}$. An interesting feature of $\CM$ is that in addition to each good $j$ having a {\em price} $p_j \geq 0$,
 and each agent $i$ has a {\em price-offset} $q_i \in \Qplus$. The {\em cost of one unit of $j$ for $i$} is defined to be $p_j + q_i$, and the {\em cost of $i$'s bundle under allocation $x$} is defined to be
$\sum_j {(p_j + q_i) x_{ij}} .$

Define the {\em price-adjusted bang-per-buck} of $i$ for good $j$ to be $ {u_{ij}} \over {p_j + q_i}$. 

\begin{definition}
\label{def.bpb}
	Given prices $p$ and price-offset $q_i$, define the {\em maximum price-adjusted bang-per-buck} of $i$ to be
	$\gamma_i = \max_j \left\{ {u_{ij} \over {p_j + q_i}} \right\} .$ 
We will say that goods in the set 
$S_i = \arg \max_j  \left\{ {u_{ij} \over {p_j + q_i}} \right\}$ 
constitute $i$'s {\em maximum price-adjusted bang-per-buck goods}.
\end{definition}

By the KKT conditions for program (\ref{eq.CP-LAD-copy}), 
an optimal allocation for $i$ must consist of goods from $S_i$. Furthermore, from KKT Condition 6 we derive the important condition:
\begin{equation}
	\label{eq.def-Gamma}
	 \gamma_i = v_i - c_i .
\end{equation}

The {\em money $m_i$ of agent $i$} is not fixed, but is a function of $v_i$, the eventual utility of $i$; the name given to this market model follows from this fact. We will define $m_i$ as follows.
	\begin{equation}
		\label{eq.money}
 m_i \ = \ {{v_i} \over {(v_i - c_i)}} \ = \  1 + {{c_i} \over {(v_i - c_i)}}   .
	\end{equation}
	
By KKT Condition 6, 
$\ \sum_{j \in G} {(p_j + q_i) x_{ij}} \ = \ {1 \over {(v_i - c_i)}}  \sum_{j \in G} {u_{ij} x_{ij}} \ = \ {{v_i} \over {(v_i - c_i)}} \ = \ m_i .$
\medskip

Hence, an {\em optimal bundle for $i$} consists of goods from $S_i$ costing $m_i$ money. Furthermore, the utility derived by $i$ from this bundle is given by  
\[  \ {m_i \cdot \gamma_i} \ = \ {{v_i} \over {(v_i - c_i)}} \cdot (v_i - c_i) \ =  v_i .\]

\begin{definition}
	\label{def.equilibrium}
	Allocation and prices $(x, p, q)$ are said to be an {\em equilibrium} for market $\CM$ if each agent receives an optimal bundle of goods and the market clears, i.e., all goods are sold and the money of agents is fully spent.
\end{definition}

\begin{remark}
	\label{rem.remarks}
By (\ref{eq.money}), the cost of $c_i$ units of utility for $i$ at equilibrium is 
${c_i \over \gamma_i} = {c_i \over {v_i - c_i}} = m_i - 1 .$
Therefore every agent spends one dollar more than the money she must spend just to accrue her disagreement utility. Hence, every agent's net utility, over and above her disagreement utility, is worth exactly one dollar. This is an aspect of the {\em fairness} of the Nash bargaining solution.

\end{remark}

We have stated above how an instance $I = (n, u, c)$ is transformed to an instance $\CM = (n, u, c)$ of a variable-budget market. By the assertions made above about money and optimal bundles, we get:

\begin{lemma}
\label{lem.reduce}
$(x, p, q)$ is a solution to instance $I$ of {\em 1DLAD} if and only if it is an equilibrium for the variable-budget market $\CM$.
\end{lemma}

\subsection{Algorithm for \emph{1DLAD} under Dichotomous Utilities}
\label{sec.mechanism-ADNB}
{\bf Notation:}
We will denote by $H = (A, G, E)$ the bipartite graph on vertex sets $A$ and $G$, and edge set $E$, with $(i, j) \in E$ iff $u_{ij} = 1$. For $A' \subseteq A$ and $G' \subseteq G$, we will denote by $H[A', G']$ the restriction of $H$ to vertex set $A' \cup G'$. If $\nu$ is a matching in $H$, $\nu \subseteq E$, and $(i, j) \in \nu$ then we will say that $\nu(i) = j$ and $\nu(j) = i$. For any subset $S \subseteq A$ ($S \subseteq G$), will denote by $N(S)$ the set of neighbors of vertices in $S$. For a set $S \subseteq G$, we will denote $\sum_{i \in S} {c_i}$ by $c(S)$.

The mechanism is given in Mechanism \ref{alg.ADNB}. 
Step 1 is executed if $H$ has a perfect matching. If so, each agent $i$ accrues unit utility, i.e., $v_i = 1$. Therefore, $m_i = {1 \over {1 - c_i}}$. We have set $p_j = 0$ for each good $j$ and $q_i = m_i$ for each agent. In Step 2, again agents in $A_2$ derive unit utility and the manner of setting $p_j, \ j \in G_2$ and $q_i, \ i \in A_2$ is similar.

We will set $q_i = 0$ for each $i \in A_1$. Let $S \subseteq G_1$ and let $T = N(S) \cap A_1$. Assume that the prices $p_j$ of all goods in $S$ are $\theta$. For $i \in T$, we can express $\gamma_i$, $m_i$ and $v_i$ in terms of $\theta$ as follows. By \eqref{eq.def-Gamma}, 
\[ \gamma_i = {1 \over \theta} \ \ \ \ \mbox{and} \ \ \ \ m_i = \ 1 + {c_i \over {(v_i - c_i)}} \ = \ 1 + {c_i \over \gamma_i} \ = \ 1 + \theta c_i .\]
Finally, 
 $v_i \ = \ m_i \cdot \gamma_i = c_i + {1 \over \theta} .$ 
We will say that set $S$ is {\em tight} if the total prices of $S$ equal the total money of $T$, i.e., $p(S) = m(T)$. Substituting in terms of $\theta$ we get that $S$ is tight if
\[ p(S) \ = \ |S| \cdot \theta \ = \ m(T) \ = \ |T| + c(T) \cdot \theta \]
Hence we get that $S$ is tight if 
\[ \theta = {{|T|} \over {|S| - c(T)}} . \] 

Step 2a finds such a vertex cover, say $G_1 \cup A_2$, where $G_1 \subseteq G$ and $A_2 \subseteq A$. Lemma \ref{lem.expand} is a variant of Lemma 10 in \cite{VY-HZ}, making a change in the first part, so as to support the new vertex cover found. 

\begin{lemma}\label{lem.expand}
The following hold:
\begin{enumerate}
    \setlength\itemsep{-0.2em}
	\item $ \forall \ S \subseteq A_2, \ |N(S) \cap G_2| > |S|. $
	\item For any set $S \subseteq G_1$, $|N(S) \cap A_1| \geq |S|$.
\end{enumerate}
\end{lemma}

\begin{proof}
For first, if $|N(S) \cap G_2| \leq |S|$ then $(G_1 \cup N(S)) \cup (A_2 - S)$ is either a smaller vertex cover or has fewer vertices from $A$, leading to a contradiction. For second, if $|N(S) \cap A_1| \leq |S|$ then $(G_1 - S) \cup (A_2 \cup N(S))$ is a smaller vertex cover, leading to a contradiction. 
\end{proof}

By Lemma \ref{lem.expand}, $|T| \geq |S|$. Furthermore, $c(T) \geq 0$. Therefore, when $S$ goes tight, $\theta \geq 1$. 

We initialize the prices of all goods in $G_1$ to $\theta = 1$ and start raising $\theta$ until a set, say $S_1$ goes tight when $\theta = \theta_1$. The prices of goods in $S_1$ are frozen at $\theta_1$ and $S_1$ and $N(S_1)$ are removed from consideration. The prices of the rest of the goods in $G_1$ are raised uniformly, starting from $\theta_1$ until another set goes tight. Computation of tight sets is done by the subroutine FLOW, which is a substantial generalization of the Simplified DPSV Algorithm~\cite{DPSV08}, and is given in Section \ref{sec.flow}. 

Assume this process finds $k$ tight sets, $S_1, \ldots S_k$, in that order when $\theta$ is $\theta_1, \theta_2 \ldots , \theta_k$, respectively. Clearly, $1 \leq \theta_1 < \theta_2 < \ldots < \theta_k$, and  
\[ \mbox{for} \ \  1 \leq j \leq k, \ \ \ \theta_j = {{|N(S_j)|} \over {|S_j| - c(N(S_j))}} . \] 

Since agents in $N(S_j)$ derive utility from goods in $S_j$, the total utility derived by them is $|S_j|$. Since this instance is guaranteed to be feasible, the total disagreement utility of $N(S_j)$ is strictly smaller. Therefore, the denominator, ${|S_j| - c(N(S_j))} > 0$. 

If $i \in N(S_l)$, then $i$ is allocated ${m_i \over \theta_j}$ amount of goods from $S_l$. Since $\theta_i \geq 1$, the amount of goods allocated to an agent is at most one unit and therefore, these will be partial allocations. As before, the rest of the allocations of these agents come from the unmatched goods of $G_2$. Note that these goods have zero prices and agents in $A_1$ derive zero utility from them.

By definition of neighborhood of sets, if $i \in N(S_l)$, then $i$ cannot have edges to $S_1, \ldots S_{l-1}$; however, $i$ can have edges to $S_{l+1}, \ldots , S_k$. Therefore, the cheapest goods from which she accrues unit utility are in $S_l$. Since $i$ has been allocated goods from this set, exhausting her money, she gets an optimal bundle. Hence, equilibrium allocations and prices have been computed.

\begin{figure}[t]
	\begin{wbox}
		\begin{alg}\label{alg.ADNB}
		{\bf Algorithm for {\em 1DLAD} under Dichotomous Utilities}
		\begin{enumerate}
			\item If $H$ has a perfect matching, say $\nu$, then do:
			 \begin{enumerate}
			 	\item $\forall i \in A: \ x_{i \nu(i)} \leftarrow 1$.
			 	\item $\forall i \in A: \ q_i \leftarrow {1  \over {1 - c_i}}$.
			 	\item $\forall j \in G$: \ $p_j \leftarrow 0$.  \ Go to Step 4. 
			 \end{enumerate}
		
			\item Else do:
			\begin{enumerate}
			\item Find a minimum vertex cover in $H$ that minimizes the number of vertices picked from $A$, say $G_1 \cup A_2$. Let $A_1 = A - A_2$ and $G_2 = G - G_1$. 
			\item Find a maximum matching in $H[A_2, G_2]$, say $\nu$. 
			\item $\forall i \in A_2: \ x_{i \nu(i)} \leftarrow 1$.
			\item $\forall i \in A_2: \ q_i \leftarrow {1  \over {1 - c_i}}$.
			 \item $\forall j \in G_2$: $p_j \leftarrow 0$. 
			\end{enumerate}

			\item
			\begin{enumerate}
			 \item $\forall i \in A_1: \ q_i \leftarrow 0$.
			\item Run subroutine FLOW on $H[A_1, G_1]$ to obtain partial allocation $x$ for agents in $A_1$ and prices $p$ for goods in $G_1$.
			\item $\forall i \in A_1$: Allocate unmatched goods of $G_2$ to satisfy the size constraint. 
			\end{enumerate}
			\item 	Output $(x, p, q)$ and Halt.

		\end{enumerate} 
		\end{alg}
	\end{wbox}
\end{figure}

\subsubsection{Subroutine FLOW}
\label{sec.flow}

This subroutine is the workhorse of Mechanism \ref{alg.ADNB}. Via flow computations, it allocates goods in $G_1$ to agents in $A_1$. Since $|G_1| < |A_1|$, these will be partial allocations. The rest of the allocations of $A_1$ come from the unallocated goods in $G_2$ which are exactly $|A_1| - |G_1|$ in number. Note that agents in $A_1$ derive zero utility from goods in $G_2$. Therefore, the entire equilibrium utility assigned to agents in $A_1$ comes from goods in $G_1$, amounting to a total of $|G_1|$ units of utility. FLOW also computes prices of goods in $G_1$.

Since the given instance of {\em 1DLAD} is feasible, $\forall i \in A$, $v_i > c_i$, where $v_i$ is the utility assigned to $i$ in the Nash bargaining solution. Therefore feasibility guarantees that $\sum_{i \in A_1} {v_i} = |G_1| > c(A_1)$. 

The high level idea of FLOW is similar to that of \cite{DPSV08}, i.e., prices of goods in $G_1$ are initialized in such a way that their equilibrium price is not exceeded. Prices are raised simultaneously and continuously until goods in some subset $G' \subseteq G_1$ attain equilibrium prices. FLOW then assigns these goods to agents who desire these goods, i.e., agents in $N(G')$. FLOW then removes $G'$ and $N(G')$ from $H[A_1, G_1]$, and proceeds with the rest of the graph in the next iteration, until the entire graph is dealt with.  

Throughout the first iteration, all goods in $G_1$ will have the same price; let $\theta$ denote the {\em current price} of goods. We will initialize $\theta$ to 1. By the second part of Lemma \ref{lem.expand}, the total price of all goods in any set $S \subseteq G_1$ is at most the money of agents in $N(S)$, i.e., agents who desire these goods. Therefore, with this initialization, we have not exceeded equilibrium prices of goods in $G_1$.

Observe that at any point in the iteration, the maximum price-adjusted bang-per-buck of agent $i \in A_1$ is $ \gamma_i = {1 / \theta}$, since $q_i = 0$. Therefore, the  {\em current money} of $i \in A_1$ is 
\[ m_i \ = \ 1 + {c_i \over \gamma_i} \ = \ 1+  c_i \theta  . \] 

FLOW increases $\theta$ continuously and as a result, in general, the money of agents also increases. This is a major departure from the Fisher market model in which the money of agents is fixed.

At any point in the iteration, the total money of all agents in $A_1$ is $|A_1| + c(A_1) \theta$ and the total value of all goods in $G_1$ is $|G_1| \theta$. Since $c(A_1) < |G_1|$, the money of agents is increasing at a slower rate than the value of goods. Therefore, the two will become equal for a finite value of $\theta$. A similar statement holds for any set $S \subseteq G_1$ and $N(S)$, which is the set of buyers interested in goods in $S$. 

At any point in the iteration, for $S \subseteq G_1$ let us define the total price of $S$ to be $p(S) = |S| \theta$. And for $T \subseteq A_1$, define the total money of $T$ to be $m(T) = |T| + c(T) \theta$. Let $\theta^*$ be the minimum value of $\theta$ at which there is $S \subseteq G_1$  such that $p(S) = m(N(S))$, i.e., the total price of $S$ equals the total money of agents who are interested in buying goods in $S$. Let $S^*$ be the maximal such set; it will be called a {\em tight set}. Lemma \ref{lem.Ss} shows how to find $\theta^*$ and $S^*$, using at most $n$ max-flow computations. It also shows how to allocate goods in $S^*$ to agents in $N(S^*)$ so as to exhaust their money. As stated above, $S^*$ and $N(S^*)$ will be removed from the graph and FLOW will move to the next iteration.
 
Using $H[A_1,G_1]$, we construct a network $\cT$ which helps us find $\theta^*$ in polynomial time. Pick $s$ and $t$ to be the source and sink of $\cT$ and add $\forall j \in G_1$, directed edges $(s, j)$ having capacity $\theta$, $\forall (i, j)$ s.t. $u_{ij} = 1$, directed edge $(j, i)$ having infinite capacity, and $\forall i \in A_1$, directed edge $(i, t)$  having capacity $1 + c_i \theta$. 

Since at the start of the iteration, $\forall S \subseteq G_1$, $p(S) \leq m(N(S))$, we have that $(\{s\}, \ G_1 \cup A_1 \cup \{t \})$ is a minimum $s-t$ cut in $\cT$. The next lemma is crucial for finding $\theta^*$. For this purpose, we will denote an arbitrary minimum $s-t$ cut in $\cT$ by $(\{s\} \cup C_1 \cup D_1, \ C_2 \cup D_2 \cup \{t\})$, where $C_1 \subseteq G_1, \ D_1 \subseteq A_1, \ C_2 = G_1 - C_1, \ D_2 = A_1 - D_1$. 

\begin{lemma}
	\label{lem.theta}
	The following hold:
	\begin{enumerate}
    \setlength\itemsep{-0.2em}
		\item If $\theta \leq \theta^*$ then $(\{s\}, \ G_1 \cup A_1 \cup \{t \})$ is a minimum $s-t$ cut in $\cT$.
		\item If $\theta > \theta^*$ then $(\{s\}, \ G_1 \cup A_1 \cup \{t \})$ is not a minimum $s-t$ cut in $\cT$. Furthermore, if $(\{s\} \cup C_1 \cup D_1, \ C_2 \cup D_2 \cup \{t\})$ is a minimum $s-t$ cut, then $S^* \subseteq C_1$.
	\end{enumerate}
\end{lemma}
\begin{proof}
	{\bf 1).} Since $\theta \leq \theta^*$, $\forall S \subseteq G_1$, $p(S) \leq m(N(S))$. Therefore, $(\{s\}, \ G_1 \cup A_1 \cup \{t \})$ is a minimum $s-t$ cut in $\cT$.
	
	{\bf 2).} At $\theta = \theta^*$, $p(S^*) = m(N(S^*))$. By feasibility, $c(N(S^*)) < |S^*|$. Therefore, as $\theta$ increases beyond $\theta^*$, $p(S^*)$ increases at a faster rate than $m(N(S^*))$. Therefore, for $\theta > \theta ^*$, $p(S^*) > m(N(S^*))$, and	hence the cut $(\{s\} \cup S^* \cup N(S^*), (G_1 - S^*) \cup (A_1 - N(S^*)) \cup \{t\})$ has a smaller capacity than $(\{s\}, \ G_1 \cup A_1 \cup \{t \})$, implying that the latter is not a minimum $s-t$ cut. 

	At $\theta > \theta^*$, let $(\{s\} \cup C_1 \cup D_1, \ C_2 \cup D_2 \cup \{t\})$ be a minimum $s-t$ cut with $S^* \cap C_1 = S_1$ and $S^* \cap C_2 = S_2$. Assume for contradiction that $S_2 \neq \emptyset$. Let $N(S_1) = T_1$. Now $T_1 \cap D_2 = \emptyset$, since otherwise there will be an infinite capacity edge crossing the min-cut. Therefore $T_1 \subseteq D_1$. Let $N(S_2) \cap D_2 = T_2$.
	
	Consider the situation when $\theta = \theta^*$. We have
	\[ {{m(N(S^*))} \over {p(S^*)}} = 1 \geq  {{m(T_1) + m(T_2)} \over {p(S_1) + p(S_2)}} .    \]
	By definition of $\theta^*$, ${{m(T_1)} \over {p(S_1)}} \geq 1$. Therefore, we get that ${{m(T_2)} \over {p(S_2)}} \leq 1$. Again by feasibility, at $\theta > \theta^*$, $m(T_2) < p(S_2)$. Therefore, moving $S_2$ into $C_1$ and $T_2$ into $D_2$ will lead to a smaller cut, contradicting the minimality of the cut being considered. Therefore $S_2 = \emptyset$ and hence $S^* \subseteq C_1$.
\end{proof}

\begin{lemma}
	\label{lem.Ss}
	$n$ max-flow computations suffice for computing	$\theta^*$ and $S^*$, and 
	also for allocating goods in $S^*$ to agents in $N(S^*)$.
\end{lemma}
\begin{proof}
	Start by finding the value of $\theta$ for which $m(A_1) = p(G_1)$, i.e., $|A_1| + \theta c(A_1) = \theta |G_1|$. By definition of $\theta^*$, $\theta \geq \theta^*$. Now $\theta = \theta^*$ if and only if the cuts $(\{s\}, \ G_1 \cup A_1 \cup \{t \})$ and $(\{s\} \cup G_1 \cup A_1, \ \{t \})$ have equal capacity. If so, both will be minimum $s-t$ cut in $\cT$  and $S^* = G_1$. A max-flow in the network will yield allocations of goods in $S^*$ to agents in $N(S^*)$.
	
	If $\theta > \theta^*$, then by Lemma \ref{lem.theta}, $(\{s\}, \ G_1 \cup A_1 \cup \{t \})$ will not be a minimum $s-t$ cut in $\cT$ and $S^* \subseteq C_1 \subset G_1$. In this case, we will iterate with the network restricted to $C_1 \cup D_1$, which has fewer goods.
\end{proof}

Observe how the feasibility of the given instance, together with Lemmas \ref{lem.theta} and \ref{lem.Ss}, help get around the ``chicken-and-egg'' problem, mentioned in the first paragraph of Section \ref{sec.mechanism-ADNB}. This problem is inherent in the variable-budget market model. 
All steps of Mechanism \ref{alg.ADNB} and FLOW, namely finding a maximum matching, a minimum vertex cover and running polynomially many max-flow computations, can be executed in strongly polynomial time. Hence we get:

\begin{theorem}
	\label{thm.ADNB}
	There is a combinatorial, strongly polynomial time algorithm for the dichotomous utilities case of the market {\em 1DLAD}.
\end{theorem}

\section{Discussion}
\label{sec.discussion}

The following question naturally arises: Is there a way of suitably modifying our algorithm for 1DLAD to obtain a mechanism that is strategyproof, as was done for the dichotomous case of HZ~\cite{Bogomolnaia2004random,Aziz20}? In Appendix~\ref{sec.DSIC-ADNB} we prove that our mechanism for 1DLAD, given in Section \ref{sec.1DLAD}, is strategyproof under the assumption that the disagreement utilities $c_i$'s are public knowledge and therefore are truthfully reported. If they can also be tampered with, we leave the problem of obtaining a strategyproof mechanism open.

Is $\epsilon$-approximate ADHZ in PPAD and is it PPAD-hard? Is program \eqref{eq.CP-LAD-copy} for {\em 1DLAD} a rational convex program for arbitrary utilities? We believe not and leave the problem of finding an instance which has only irrational optimal solutions. Note though that the variable-budget market model holds for arbitrary utilities. Is there an efficient algorithm for computing an approximate equilibrium for this model? 

\section{Acknowledgement} \label{sec.ack}

We wish to thank Federico Echenique for valuable discussions.

\bibliographystyle{abbrv}
\bibliography{refs,references}

\appendix
\section{Preliminaries}
\subsection{Ordinal vs Cardinal Utilities}\label{sec.ord-car}
 
The two ways of expressing utilities of goods -- ordinal and cardinal -- have their own pros and cons and neither dominates the other. On the one hand, the former is easier to elicit from agents and on the other, the latter is far more expressive, enabling an agent to not only report if she prefers good $A$ to good $B$ but also by how much. \cite{Abdulkadirouglu-Cardinal} exploit this greater expressivity of cardinal utilities to give mechanisms for school choice which are superior to ordinal-utility-based mechanisms. 

The following example illustrates the advantage of cardinal vs ordinal utilities. The instance has three types of goods, $T_1, T_2, T_3$, and these goods are present in the proportion of $(1\%, \ 97\%, \ 2\%)$. Based on their utility functions, the agents are partitioned into two sets $A_1$ and $A_2$, where $A_1$ constitute $1\%$ of the agents and $A_2$, $99\%$. The utility functions of agents in $A_1$ and $A_2$ for the three types of goods are $(1, \ \epsilon, \ 0)$ and $(1, \ 1- \epsilon, \ 0)$, respectively, for a small number $\epsilon > 0$. The main point is that whereas  agents in $A_2$ marginally prefer $T_1$ to $T_2$, those in $A_1$ overwhelmingly prefer $T_1$ to $T_2$. 

Clearly, the ordinal utilities of all agents in $A_1 \cup A_2$ are the same. Therefore, a mechanism based on such utilities will not be able to make a distinction between the two types of agents. On the other hand, the HZ mechanism, which uses cardinal utilities, will fix the price of goods in $T_3$ to be zero and those in $T_1$ and $T_2$ appropriately so that by-and-large the bundles of $A_1$ and $A_2$ consist of goods from $T_1$ and $T_2$, respectively.

\subsection{The Fisher Market Model}\label{sec:fisher}
The {\em Fisher market model} consists of a set $A = \{1, 2, \ldots n\}$ of agents and a set $G = \{1, 2, \ldots, m\}$ of infinitely divisible goods. By fixing the units for each good, we may assume without loss of generality that there is a unit of each good in the market. Each agent $i$ comes to the market with a budget $b_i$ and has utilities $u_{ij} \geq 0$ for each good $j$. In the case of linear utilities, agent $i$'s utility from allocation $(x_{ij})_{j\in G}$ is $\sum_j u_{ij}x_{ij}$. 

\begin{definition}
    A \emph{Fisher equilibrium} is a pair $(x, p)$ consisting of an \emph{allocation} $(x_{ij})_{i \in A, j \in G}$ and \emph{prices} $(p_j)_{j \in G}$ with the following properties.
    \begin{enumerate}
        \item Each agent $i$ spends at most their budget, i.e., $\sum_{j \in G}{p_j x_{ij}} \leq b_i$.
        \item Each agent $i$ gets an \emph{optimal bundle}, i.e., utility maximizing bundle at prices $p$. Formally:
 {\small           \[
                \sum_{j \in G}{u_{ij} x_{ij}} = \max \left\{\sum_{j \in G}{u_{ij} y_j} \;\middle|\; y \in \bbR_{\geq 0}^G, \sum_{j \in G}{p_j y_j } \leq b_i\right\}.
            \]}
        \item The market clears, i.e., each good with positive price is fully allocated to the agents.
    \end{enumerate}
\end{definition}

The set of equilibria of a linear Fisher market corresponds to the set of optimal solutions of the following Eisenberg-Gale convex program \cite{eisenberg}.
\begin{maxi*}
    {(x_{ij})_{i \in A, j \in G}}
    {\sum_{i \in A}{b_i \log \sum_{j \in G}{u_{ij} x_{ij}}}}
    {}
    {}
    \addConstraint{\sum_{i \in A}{x_{ij}}}{\leq 1}{\quad \forall j \in G}
    \addConstraint{x_{ij}}{\geq 0}{\quad \forall i \in A, j \in G.}
\end{maxi*}

This is a \emph{rational convex program} and in fact it motivated the definition of this concept~\cite{va.rational}.

\subsection{Nash Bargaining Problem} 
\label{sec.Nash}

An {\em $n$-person Nash bargaining problem} consists of a pair $(\CN, c)$, where $\CN \subseteq \R_+^n$ is a compact, convex set and $c \in \CN$. The set $\CN$ is called the {\em feasible set} -- its elements are vectors whose components are utilities that the $n$ players can simultaneously accrue. Point $c$ is the {\em disagreement point} -- its components are utilities players accrue if they decide not to participate in the proposed solution. 

The set of $n$ agents will be denoted by $A$ and the agents will be numbered $1, 2, \ldots n$. Instance $(\CN, c)$ is said to be {\em feasible} if there is a point in $\CN$ at which each agent does strictly better than her disagreement utility, i.e., $\exists \vv \in \CN$ such that $\forall i \in A, \ v_i > c_i$, and {\em infeasible} otherwise. In game theory it is customary to assume that the given Nash bargaining problem $(\CN, c)$ is feasible; we make this assumption as well.

The solution to a feasible instance is the point $\vv \in \CN$ that satisfies the four axioms: 

\begin{enumerate}
\item
{\em  Pareto optimality:}  No point in $\CN$ weakly dominates $\vv$.
\item
{\em  Symmetry:} If the players are renumbered, then a corresponding renumber the coordinates of $\vv$ is a solution to the new instance.
\item
{\em  Invariance under affine transformations of utilities:} If the utilities of any player are redefined by
multiplying by a scalar and adding a constant, then the solution to the transformed problem is obtained by
applying these operations to the particular coordinate of $\vv$.
\item
{\em  Independence of irrelevant alternatives:} If $\vv$ is the solution to $(\CN, c)$, and 
$\CS \subseteq \R_+^n$ is a compact, convex set satisfying $c \in \CS$ and  $\vv \in \CS \subseteq \CN$, then $\vv$ is also the solution to $(\CS, c)$.
\end{enumerate}

Via an elegant proof, Nash proved:

\begin{theorem}
[ Nash \cite{Nash1953two}]
\label{thm.nash}
If problem $(\CN, c)$ is feasible
then there is a unique point in $\CN$ satisfying the axioms stated above. Moreover, this point is obtained by maximizing $\Pi_{i \in A}  {(v_i - c_i)}$ over $\vv \in \CN$.
\end{theorem}

Nash's solution to his bargaining problem involves maximizing a concave function over
a convex domain, and is therefore the optimal solution to the following convex program.

	\begin{maxi}
		{} {\sum_{i \in A}  {\log (v_i - c_i)}}
			{\label{eq.CP-Nash}}
		{}
		\addConstraint{}{\vv \in \CN}
	\end{maxi}

As a consequence, if for a specific game, a separation oracle can be implemented
in polynomial time, then using the ellipsoid algorithm one can get as good an approximation to
the solution of this convex program as desired in time polynomial in the number of bits
of accuracy needed \cite{GLS}. In this paper, we do better by showing that most of our programs  are rational convex programs and hence we can get exact solutions in polynomial time.

\subsection{Rational Convex Program}\label{sec.RCP1}
The notion of a rational convex program (RCP) was motivated by the remarkable program given by Eisenberg and Gale \cite{eisenberg}. Its optimal solution gives an equilibrium allocation to the linear Fisher market and its dual gives equilibrium prices. If all parameters of the linear Fisher market are rational numbers, then there is always an equilibrium consisting of rational numbers; this obviously carries over to the Eisenberg-Gale convex program as well. The surprising aspect is that this happens despite the fact that its objective function is non-linear, consisting of logarithms! 

\begin{definition}
	\label{def.RCP} 
A nonlinear convex program is said to be a {\em rational convex program (RCP)} if for any setting of its parameters to rational numbers such that there is a finite optimal solution, it admits an optimal solution consisting of rational numbers. Moreover, the solution can be written using polynomially many bits in the number of bits needed to write all the parameters. 
\end{definition}

The significance of this notion lies in that the exact optimal solution to such a program can be found in polynomial time using the ellipsoid algorithm and simultaneous Diophantine approximation \cite{GLS, Jain-AD}. In contrast, an arbitrary nonlinear convex program admits an ellipsoid-based approximation algorithm whose running time is polynomial in the size of the input and $\log{1/\epsilon}$, where $\epsilon$ is the precision desired \cite{GaleS, Vishnoi.book}.

Since this notion is not well known, it will be useful to draw an analogy with an {\em integral linear program (ILP)}, which is an LP that always has integral optimal solutions. Some of the central problems of the field of combinatorial optimization, such as matching, flow and minimum spanning tree, share the feature that they possess LP-relaxation. A feasible LP with rational parameters always has a rational optimal solution. Hence the solution produced by an ILP is qualitatively better in than that it is integral. Needless to say this distinction is of central importance, since it enables the use of an LP-solver to solve these combinatorial optimization problems, whose solutions are necessarily integral. Analogously, the solution to an RCP is qualitatively better than that of an arbitrary nonlinear convex program, since the latter will typically have only irrational solutions.

\section{Rational Convex Programs}
\label{sec.RCP-App}

In Appendix~\ref{sec.RCP-1LAD}, we prove Theorem~\ref{thm.RCP-1} to show that \eqref{eq.CP-LAD-copy} is an RCP for the model \emph{1DLAD}. In Appendix~\ref{asec.rcp}, we present an RCP for the HZ scheme under dichotomous utilities. 
\subsection{Proof of Theorem~\ref{thm.RCP-1}}\label{sec.RCP-1LAD}

\begin{proof}
Let $(x, p, q)$ be an optimal solution to program \eqref{eq.CP-LAD-copy}. We will say that $p_j$ is the {\em price} of good $j$. Assume that $x_{ij} > 0$. There are two cases:

{\bf Case (a), $u_{ij} = 0$:} By KKT Condition (6), 
\[ u_{ij} = (v_i - c_i) (p_j + q_i) = 0 .\]
Since $v_i - c_i > 0$, we get that $p_j = q_i = 0$.

{\bf Case (b), $u_{ij} = 1$:} By KKT Conditions (5) and (6), among goods for which $i$ has utility 1, $p_j$ must be smallest. Furthermore, by KKT Condition (6), 
\[ u_{ij} = (v_i - c_i) (p_j + q_i) = 1 .\]

Therefore, 
\[ \sum_{j \in G} {u_{ij} x_{ij}} = (v_i - c_i) \sum_{j \in G} {(p_j + q_i) x_{ij}} .    \]
If $q_i > 0$, by KKT Condition (4), $\sum_{j \in G} {x_{ij}} = 1$, giving $\sum_{j \in G} {q_i  x_{ij}} = q_i$. If $q_i = 0$, we still get that $\sum_{j \in G} {q_i  x_{ij}} = 0 = q_i$. Therefore we have:
\[ v_i = \sum_{j \in G} {u_{ij} x_{ij}} =  (v_i - c_i) (\sum_{j \in G} {p_j x_{ij}} + q_i) .    \]

Hence,
\[ \sum_{j \in G} {p_{j} x_{ij}} =  {v_i \over {(v_i - c_i)}} - q_i  .    \]

If for agent $i$, $q_i > 0$, then by KKT Condition (4), $\sum_{j \in G} {x_{ij}} = 1$, i.e., $i$'s allocation consists of one unit of goods, of minimum cost, for which $i$ has utility 1. Assign one full unit of such a good, say $j$, to $i$ and remove $i$ and $j$ from consideration.

Therefore, in the remaining problem, for every agent $i$, $q_i = 0$. Therefore we have:
\[ \sum_{j \in G} {p_{j} x_{ij}}  =  {v_i \over {(v_i - c_i)}} = 1 +  {c_i \over {(v_i - c_i)}} .   \]

Suppose $x_{ij} > 0$ and $u_{ij} = 1$. Then we have
\[ u_{ij} = 1 = (v_i - c_i) p_j \ \implies \ p_j = {1 \over {(v_i - c_i)}} .\]

Let $G' \subseteq G$ be the set of goods whose price is positive, i.e., $G' = \{j \in G \ | \ p_j > 0 \}$. Let $A' \subseteq A$ be the agents who are allocated goods in $G'$. Consider a bipartite graph on vertex sets $G'$ and $A'$ with an edge $(i, j)$ for $i \in A', \ j \in G'$ if $x_{ij} = 1$. Let $C$ be a connected component in this graph and let the set of goods and agents in $C$ be $G_C'$ and $A_C'$, respectively. Since agents buy only cheapest goods for which they have utility 1, all prices of goods in $G_C'$ are equal, say they are $p_C$. The total money of agents in $A_C'$ is
\[ \sum_{i \in A_C'} \sum_{j \in G_C'} {p_j x_{ij}} =  \sum_{i \in A_C'} {1 + {c_i \over {(v_i - c_i)}}} = |A_C'| + c(A_C') p_C , \]   
where $c(A_C')$ is the sum of disagreement utilities of all agents in $A_C'$. i.e., $c(A_C') = \sum_{i \in A_C'} {c_i}$. The latter is a rational number by assumption. 

On the other hand, 
\[ \sum_{j \in G_C'} \sum_{i \in A_C'} {p_j x_{ij}} =  \sum_{j \in G_C'} {p_j} = |G_C'| p_C .  \]

Equating the two we get
 $|G_C'| p_C = |A_C'| + c(A_C') p_C .$ 
Solving this, we get that $p_C$ is a rational number. Using the equation 
 $p_j = {1 \over {(v_i - c_i)}}$ 
we get that $v_i$ is a rational number for each $i \in G_C'$. Since $v_i$ is the total allocation of utility 1 goods to $i$, we conclude that the entire allocation is rational.

Finally let $G_0$ and $A_0$ be the set of goods which are not fully allocated and agents who received less than a full unit of allocations. For $i \in A_0, \ j \in G_0$, $u_{ij} = 0$, since otherwise making such an allocation, the objective function value of \eqref{eq.CP-LAD-copy} can be increased. Make rational allocations from $G_0$ to $A_0$ so that each agent gets a full unit of allocation. This is the solution of 1DLAD under dichotomous utilities. 
\end{proof}

Since the objective function in \eqref{eq.CP-LAD-copy} is strictly concave, the utility derived by each agent $i$ must be the same in all solutions of this convex program. Hence, we get the following corollary. It can be seen as a variant of the well-known \emph{Rural Hospital Theorem}; see \cite{GusfieldI} for the latter. 

\begin{corollary}
\label{cor.concave}
Each agent, $i$, gets the same utility, $v_i$, under all optimal solutions to \eqref{eq.CP-LAD-copy} under dichotomous utilities.
\end{corollary}

\subsection{An RCP for the HZ Scheme under Dichotomous Utilities}\label{asec.rcp}

We will assume without loss of generality that each agent $i\in A$ likes some good $j\in G$, i.e.\ $u_{ij} = 1$. We will show that program (\ref{cp}) given below is the required RCP. Observe that it is an extension of the Eisenberg-Gale convex program~\cite{eisenberg} via the second constraint, i.e., the amount of goods allocated to each agent is at most 1. 

\begin{equation}\label{cp}
\begin{array}{ccc}
     & \max & \sum_{i\in A} \log{\sum_{j\in G} u_{ij}x_{ij}} \\
\mbox{subject to}  & \forall j\in G:   & \sum_{i\in A} x_{ij} \le 1 \\
& \forall i\in A: & \sum_{j\in G} x_{ij} \le 1 \\
& \forall i\in A,j\in G: & x_{ij} \ge 0
\end{array}\enspace .
\end{equation}

Let $p_j$'s and $\alpha_i$'s denote the non-negative dual variables for the first and second  constraints, respectively. 

\begin{theorem}
\label{thm.dichotomous}
 Any HZ equilibrium is an optimal solution to (\ref{cp}), and every optimal solution of~\eqref{cp} can be trivially extended to an HZ equilibrium. Furthermore, the latter can be expressed via rational numbers whose denominators have polynomial, in $n$, number of bits, thereby showing that \eqref{cp} is a rational convex program.
\end{theorem}

\begin{proof}
Let $u_i := \sum_{j\in G} u_{ij}x_{ij}$. Clearly, in any HZ equilibrium, since each agent $i$ is allocated an optimal bundle of goods, she will be allocated a non-zero amount of a unit-utility good and hence will satisfy $u_i > 0$. Furthermore, in an optimal solution $x$ of (\ref{cp}), every agent must have positive utility, because otherwise the objective function value will be $-\infty$. Therefore, $\forall i \in A: \ u_i > 0$. 

The KKT conditions of this program are:

\begin{enumerate}
    \setlength\itemsep{0em}
\item  $\forall i\in A:$  $\alpha_i \geq 0$. 
\item  $\forall j\in G:$  $p_j \geq 0$.
\item $\forall i\in A:$ If $\alpha_i > 0$ then $\sum_j x_{ij} = 1$. 
\item $\forall j\in G:$ If $p_j > 0$ then $\sum_i x_{ij} = 1$. 
\item 
$\forall i\in A, j\in G:  u_{ij} \le u_i (p_j + \alpha_i)$.
\item
$\forall i\in A, j\in G: \ \ x_{ij} > 0  \Rightarrow  u_{ij} = u_i (p_j + \alpha_i).$

\end{enumerate}

To prove the forward direction of the first statement, let $(x, p)$ be an HZ equilibrium. Since $x$ is a fractional perfect matching on agents and goods, it satisfies the constraints of (\ref{cp}) and is hence a feasible solution for it. We are left with proving optimality.

The KKT conditions 2, 3 and 4 are clearly satisfied by $(x, p)$. Next, consider agent $i$. If there is a good $j$ such that $p_j \leq 1$ and $u_{ij} = 1$, then $i$ will be allocated one unit of the cheapest such goods. Assume the price of the latter is $p$. Define $\alpha_i = 1- p$. Clearly $u_i = 1$. Now, it is easy to check that Conditions 1, 5 and 6 are also holding. 

Next assume that every good $j$ such that $u_{ij} = 1$ has $p_j > 1$ and let $p$ be the cheapest such price. Clearly, $i$'s optimal bundle will contain $1/p$ amount of these goods, giving her total utility $1/p$. Since the equilibrium always has a zero-priced good, that good, say $j$, must have $u_{ij} = 0$. Now, $i$ must be buying such zero-utility zero-priced goods to get to one unit of goods. We will define $\alpha_i = 0$. Again, it is easy to check that Conditions 1, 5 and 6 are holding. Hence we get that $(x, p)$ is an optimal solution to (\ref{cp}).

Next, we prove the reverse direction of the first statement. Let $(x, p)$ be an optimal solution to (\ref{cp}). Assume that agent $i$ is allocated good $j$, i.e.\ $x_{ij}>0$. We consider the following two cases:

\begin{itemize}
    \item[$(a)$] $u_{ij} = 0$. Using Condition 6 and $u_i > 0$, we get that $p_j = \alpha_i = 0$. 
    \item[$(b)$] $u_{ij} = 1$. Using Conditions 5 and 6 and $u_i > 0$, we get that the price of good $j$ is the cheapest among all goods for which $i$'s utility is 1.
\end{itemize}

For each agent $i$, multiply the equality in Condition 6 by $x_{ij}$ and sum over all $j$ to get:
\[\sum_j {x_{ij} u_{ij}} \ = \ u_i \sum_j x_{ij}(p_j + \alpha_i)\]

After canceling $u_i$ from both sides we obtain
\[ \ \sum_j x_{ij}(p_j + \alpha_i) = 1 = \sum_j {x_{ij} p_j}  + \alpha_i \sum_j {x_{ij}} .\]
Now, if $\alpha_i > 0$, then $\sum_j {x_{ij}} = 1$ and if $\alpha_i = 0$, then $\alpha_i \sum_j {x_{ij}} = 0 = \alpha_i$. Therefore, in both cases $\alpha_i \sum_j {x_{ij}} = \alpha_i$.
Hence, 
\begin{equation}
	\label{eq.alpha}
\sum_j x_{ij}p_j = 1 -\alpha_i.
\end{equation}

We will view the dual variables $p$ of the optimal solution $(x, p)$ as prices of goods. The above statement then implies that agent $i$'s bundle costs $1 - \alpha_i$.

Let $S$ denote the set of agents who get less than one unit of goods, i.e.\ $S:=\{i \in A\ |\ \sum_j x_{ij } < 1\}$, and let $T$ denote the set of partially allocated goods, i.e.\ $T:=\{j\in G\ |\ \sum_i x_{i j} < 1\}$. By Condition 4,  $p_j = 0$ for each $j\in T$. Observe that if for $i \in S$ and $j \in T$, $u_{ij} = 1$, then by allocating a positive amount of good $j$ to $i$, the objective function value of program (\ref{cp}) strictly increases, giving a contradiction. Therefore, $u_{ij} = 0$. 

Since the number of agents equals the number of goods, the total deficiency of agents in solution $x$ equals the total amount of unallocated goods. Therefore, we can arbitrarily allocate unallocated goods in $T$ to deficient agents in $S$ so as to obtain a fractional perfect matching, say $x'$. Clearly, $(x', p)$ is still an optimal solution to \eqref{cp} and is also an HZ equilibrium.

For the second statement, we will start with this solution $(x', p)$. Let $G'\subseteq G$ denote the set of goods with prices bigger than 1, i.e.\ $G'=\{j\in G\ |\ p_j > 1\}$ and let $A' \subseteq A$ be the set of agents who have allocations from $G'$. By Cases $(a)$ and $(b)$, for each $i \in A'$, there is a $j \in G'$ such that $u_{ij} = 1$; moreover this is the cheapest good for which $i$ has utility 1. We first show that each agent $i \in A'$ satisfies $\alpha_i = 0$. If $\sum_{j \in G} {x_{ij}} < 1$, this follows from KKT Condition 3. Otherwise, there exists $j \in G$ such that $x_{ij} > 0$ and $u_{ij} = 0$. The last statement follow from the fact that $\sum_j x_{ij}p_j \leq 1$, which follows from (\ref{eq.alpha}). Again, by Case $(a)$, $\alpha_i = 0$. Now, by (\ref{eq.alpha}), the money spent by each agent in $A'$ is exactly 1 dollar on goods in $G'$.  

Consider the connected components of bipartite graph $(A', G', E)$, where the set $E = \{(i,j) \in (A', G')\ |\ x_{ij} > 0\}$. Cases $(a)$ and $(b)$ imply that all goods in a connected component $C$ must have the same price, say $p_C$. Clearly, the sum of prices of all goods in $C$ equals the total money of agents in $C$; the latter is simply the number of agents in $C$. This implies that $p_C$ is rational. Clearly, there is a rational allocation of $1/p_C$ amount of goods to every agent in $C$.

Let $i \in A$ such that the cheapest good for which $i$ has utility 1 has price 1. If $\alpha_i = 0$, by (\ref{eq.alpha}), $i$ buys 1 dollar, and hence 1 unit, of such goods. If $\alpha > 0$, by KKT Condition 3, $\sum_{j \in G} {x_{ij}} = 1$ and therefore again $i$ has bought 1 unit of such goods. Now, without loss of generality, we will assign to $i$ an entire unit of one such good.

Finally, let $G'' \subseteq G$ denote the set of goods with prices in the interval $(0, 1)$, i.e.\ $G''=\{j\in G\ |\ 0 < p_j < 1\}$ and let $A'' \subseteq A$ be the set of agents who have allocations from $G''$. Let $i \in A''$. Since $\sum_j x_{ij}p_j < 1$, by (\ref{eq.alpha}) $\alpha > 0$. Therefore each agent in $A''$ buys one unit of goods from $G''$. Hence the allocation of goods from $G''$ to $A''$ forms a fractional perfect matching on $(G'', A'')$. Therefore, we can pick any perfect matching consistent with this fractional perfect matching and allocate goods from $G''$ integrally to $A''$. 

Hence in all cases, the allocation consists of rational numbers, completing the proof. 
\end{proof}

\begin{remark}
	\label{rem.dual}
	The proof of Theorem \ref{thm.dichotomous} shows that for the dichotomous case, the dual of (\ref{cp}) yields equilibrium prices. In contrast, for arbitrary utilities, there is no known mathematical construct, no matter how inefficient its computation, that yields equilibrium prices. In a sense, this should not be surprising, since there is a polynomial time algorithm for computing an equilibrium for the dichotomous case \cite{VY-HZ}.  
\end{remark}

Since the objective function in~\eqref{cp} is strictly concave, the utility derived by each agent $i$ must be the same in all solutions of~\eqref{cp}. Hence, we get the following corollary which can be seen as a variant of the well-known \emph{Rural Hospital Theorem}; see \cite{GusfieldI} for the latter. 

\begin{corollary}
Each agent gets the same utility under all HZ equilibria with dichotomous utilities.
\end{corollary}

\section{Strategyproofness for Mechanism \ref{alg.ADNB}}\label{sec.DSIC-ADNB}

For $i' \neq i$, let $u_{i'j} \in \{0, 1\}$, for $j \in G$, be the dichotomous utility function reported by agent $i'$, and let $u_{ij} \in \{0, 1\}$, for $j \in G$, be $i$'s true dichotomous utility function. Assume that $i$ misreports her utility function as $u'_{ij} \in \{0, 1\}$, for $j \in G$. In this section, we show that Mechanism~\ref{alg.ADNB} is strategyproof under the assumption that the disagreement utilities $c_i$'s are public knowledge. We begin with the following lemma. 

\begin{lemma}
	\label{lem.subset-MMNB}
	Consider the subsets $S_1, \ldots , S_k$ of $G_1$ in graph $H[A, G]$ which successively go tight when the prices of goods in them are $\theta_1 < \theta_2 < \ldots < \theta_k$, respectively. Let $S \subseteq S_l$ and $T \subseteq (S_{l+1} \cup \ldots \cup S_k)$. Then, 
	\begin{enumerate}
		\item  ${{|N(S)|} \over {|S| - c(N(S))}} \geq \theta_l$. 
		\item ${{|N(T)|} \over {|T| - c(N(T))}} \geq \theta_{l+1}$. 
	\end{enumerate}
\end{lemma}

\begin{proof}
	The proof of both statements is by contradiction. If the first assertion were not holding, then $S$ would freeze before $S_l$ and if the second assertion were not holding, then $T$ would freeze before $S_{l+1}$. In either case, the set of successive tight sets would be different from the one claimed. 
\end{proof}

We show the following theorem where $x$ and $x'$ be the allocations computed by our mechanism in the $u$-run and $u'$-run, respectively. 
\begin{theorem}
	\label{thm.IC-MMNB}
	Agent $i$ does not accrue more utility under allocation $x'$ as compared to $x$, when evaluated by her true utility function, i.e., 
	\[ u_i(x) \geq u_i(x') .\]
\end{theorem}

\begin{proof}
	The maximum utility an agent can accrue is 1. Therefore, if $u_i(x) = 1$, the theorem holds vacuously. Hence we will study the case that $u_i(x) < 1$, i.e., $H[A, G]$ does not have a perfect matching and $i \in A_1$. Furthermore, borrowing the terminology set up 
earlier,	
	assume that the run of FLOW partitions $G_1$ into sets $S_1, \ldots , S_k$ and these sets go tight when the prices of goods in them are $\theta_1, \theta_2 \ldots \theta_k$, respectively, with $\theta_1 < \ldots < \theta_k$. 
	
Assume that $i \in N(S_l)$ in $H[S, G]$. Then, $i$ has no edges to $S_1, \ldots, S_{l-1}$, she  has edges to $N(S_l)$ and may also have edges to $S_{l+1}, \ldots , S_k$. Since $u_i(x) = {1 \over \theta_l} < 1$, it must be the case that $|N(S_l)| > |S_l| - c(N(S_l))$ and $\theta_l > 1$. Note that the partitions $(A_1, A_2)$ of $A$ and $(G_1, G_2)$ of $G$ were obtained when we found a minimum vertex cover in the graph $H[A, G]$ which contained only the old edges of $i$. 
		
Clearly, in the $u'$-run, only the flow sent to $i$ on old edges counts towards the ``actual'' utility of $i$, i.e., counts towards $u_i(x')$. We next analyze $u_i(x')$ under the different ways in which $i$ could have manipulated her utility function.
			
	{\bf Case 1:} Agent $i$ adds new edges to $G_2$, i.e., \ $\exists j \in G_2$ s.t. $(i, j)$ is a new edge. \\
	 By Lemma \ref{lem.one-bigger} stated below, the size of a minimum vertex cover in $H'[A, G]$ is bigger than that in $H[A, G]$ by one. The new minimum vertex cover will depend on whether $i$ has picked one new edge to $G_2$ or more. Observe that none of these new edges are  covered by the old vertex cover, since they run between $A_1$ and $G_2$.
	 
	 In the first case, since we are picking a minimum vertex cover that minimizes the number of vertices picked from $A$, $j$ will be added to the cover, hence covering $(i, j)$. The new cover is $A_2 \cup G_1'$, where $G_1' = G_1 \cup \{j\}$. Now, for $S = \{j\}$, $N(S) = \{i\}$. Therefore, $S$ will go tight when $\theta = 1/(1-c_i)$ and $i$ will be allocated one unit of $j$, resulting in $u_i(x') = 0$.
	 
	 In the second case, $i$ will be added to the cover. Therefore, the new cover will be $A_2' \cup G_1$ where $A_2' = A_2 \cup \{i\}$. Therefore, $i$ will be matched in Step 3 via one of the new edges. Hence, again $u_i(x') = 0$.
\medskip

	\noindent	{\bf Case 2:} Agent $i$ does not add any new edges to $G_2$.\\
		In the $u$-run, the utility received by $i$ is $c_i + {1 \over \theta_l}$.
		
	Assume that in the $u'$-run, a set $T \subseteq G_1$ goes tight, with $i \in N'(T)$. Let $\theta$ be the prices of goods in $T$ when it goes tight. If $\theta \geq \theta_l$, then the utility received by $i$ in the $u'$-run is $c_i + {1 \over \theta} \leq c_i + {1 \over \theta_l}$ and therefore $u_i(x) \geq u_i'(x') \geq u_i(x')$. Hence, the statement of the theorem is true.
	
	Next, assume that $\theta < \theta_l$. Consider the partition of $T$ into $(T_1, T_2, T_3)$ where the three sets in this partition are $T \cap (S_1 \cup \ldots \cup S_{l-1})$, $T \cap S_l$ and $T \cap (S_{1+1} \cup \ldots \cup S_k)$, respectively. In graph $H[A, G]$, denote $N(T)$ by $V$ and define its partition $(V_1, V_2, V_3)$ as follows: The sets $V_1, V_2$ and $V_3$ are $N(T_1)$, $N(T_2) - V_1$ and $N(T_3) - (V_1 \cup V_2)$, respectively. Observe that $i \in V_2$ and that the only difference between $N(T)$ and $N'(T)$ can be $i$. However, since $i \in N'(T)$, we have that $N'(T) = N(T) = V$. 
	
	 Regarding $T$, we will first prove that $T_3 = \emptyset$. By Lemma \ref{lem.subset-MMNB}, ${{|N(T_3)|} \over {|T_3| - c(N(T_3))}} \geq \theta_{l+1} > \theta_l$. Since $N(T_3) = N'(T_3) = V_3$, the previous inequality holds in $H'[A, G]$ as well, i.e., ${{|N'(T_3)|} \over {|T_3| - c(V_3)}} > \theta_l$.
	 Furthermore, since 
	 \[ {{|V|} \over {|T| - c(V)}} = {{|V_1| + |V_2| + |V_3|} \over {(|T_1| - c(V_1)) + (|T_2| - c(V_2)) + (|T_3| - c(V_3))}} =    \theta < \theta_l , \]
	 we get ${{|V_1| + |V_2|} \over {(|T_1| - c(V_1)) + (|T_2| - c(V_2))}} < \theta$. Clearly, in $H'[A, G]$, $N'(T_1 \cup T_2) = V_1 \cup V_2$. Therefore, $(T_1 \cup T_2)$ will go tight before $T$, leading to a contradiction.
	 
	Next, we prove that $T_1 \neq \emptyset$. Suppose not, then $T \subseteq S_l$. Since $N(T) = N'(T) = V$, by Lemma \ref{lem.subset-MMNB} we get  
	\[ {{|N(T)|} \over {|T| - c(N(T))}} = {{|N'(T)|} \over {|T| - c(N(T))}} \geq \theta_l . \] 
	This contradicts the fact that $T$ goes tight at $\theta < \theta_l$ in the $u'$-run.

	We are now ready to prove that $u_i(x) \geq u_i(x')$. Let $V_2^* = (V_2 - \{i\})$. The total utility accrued by agents in $V_2$ in the $u$-run is 
\[\left(c_i + {1 \over \theta_l} \right) + \left( c(V_2^*) + {{|V_2^*|} \over \theta_l} \right) = |T_2| , \] 
	where we have split $V_2$ into $V_2^*$ and $i$, and we have used the fact that goods in $T_2$ provide all utility to agents in $V_2$ in the $u$-run.
	
	In the $u'$-run, the set $T$ freezes when the prices of goods in it are 
	\[ \theta := {{|V_1| +  |V_2|} \over {(|T_1| - c(V_1)) + (|T_2| - c(V_2))}}  , \] 
	and the total utility accrued by agents in $V_2^*$ in the $u'$-run is $c(V_2^*) + {{|V_2^*|} \over \theta}$.
	
	Since these agents don't have any edges to $S_1 \cup \ldots \cup S_{l-1}$, all this utility comes from $T_2$. Let the utility accrued by $i$ from goods in $T_2$ in the $u'$-run be $c_i + a$. Then we have,
	
\[\left(c_i + a \right) + \left( c(V_2^*) + {{|V_2^*|} \over \theta} \right) = |T_2|. \] 	
Equating with the equation for the $u$-run we get
\[ a = {1 \over \theta_l}  +  {|V_2^*|} \left({1 \over \theta} - {1 \over \theta_l} \right) .\]
Since $\theta < \theta_l$ we get that $a < {1 \over \theta_l}$.	Therefore, $ u_i(x') < u_i(x)$. 

This completes the proof of Theorem \ref{thm.ADNB}.
\end{proof}

Finally we note that Footnote~\ref{footnote1} holds for this setting as well. As a result, Mechanism \ref{alg.ADNB} computes an optimal solution for {\em 1DLAD} under bi-valued utility functions as well and is strategyproof.
	
\begin{lemma}
	\label{lem.one-bigger}
		 In Case 1 in the proof of Theorem \ref{thm.IC-MMNB}, the size of a maximum matching in $H'[A, G]$ is bigger than that in $H[A, G]$ by one. 
\end{lemma}

\begin{proof}
		 First match the edge $(i, j)$ and remove its vertices from the graph $H'[A, G]$. Let us denote the remaining graph by $H'[A, G] - \{i, j\}$. This graph is identical to the graph obtained by removing $i$ and $j$ from $H[A, G]$, which we will denote by $H[A, G] - \{i, j\}$. This follows from the fact that $H'[A, G]$ and $H[A, G]$ differ only in the edges incident at $i$. 

		 We claim that in $H[A, G] - \{i, j\}$, \ $G_1$ an $A_2$ will satisfy Hall's condition and hence both can be fully matched. The lemma then follows, since the size of a maximum matching in $H[A, G]$ equals the size of a minimum vertex cover in it which is $|G_1| + |A_2|$.
		 
		 To prove the claim, let us first consider $A_2$. As stated above, in the $u$-run, Step 2a of Algorithm \ref{alg.ADNB} finds a minimum vertex cover in $H[A, G]$ which minimizes the number of vertices picked from $A$, and as a result, the first assertion of Lemma \ref{lem.expand} changes to:
\[ \forall \ S \subseteq A_2, \ |N(S) \cap G_2| > |S|. \]
In going from $H[A, G]$ to $H[A. G] - \{i, j\}$, only one vertex is removed from $G_2$, namely $j$. Therefore, in $H[A, G] - \{i, j\}$ we have
\[ \forall \ S \subseteq A_2, \ |N(S) \cap G_2| \geq |S|. \]

Next we consider $G_1$. In the graph $H[A, G]$, by Lemma \ref{lem.expand} we have
\[ \forall S \subseteq G_1, \ |N(S) \cap A_1| \geq |S| .\]
As stated at the beginning of the proof of Theorem \ref{thm.IC-MMNB}, $|N(S_l)| > |S_l| - c(N(S_l))$ and $\theta_l >1$. By Lemma \ref{lem.subset-MMNB}, 
\[ \forall S \subseteq S_l, \ |N(S)| \geq \theta_l (|S| - c(N(S))), \ \ \mbox{therefore} \ \ |N(S)| > |S| - c(N(S)). \]

In going from $H[A, G]$ to $H[A, G] - \{i, j\}$, only one vertex is removed from $A_1$, namely $i$. Therefore, in graph $H[A, G] - \{i, j\}$ we have
\[ \forall S \subseteq S_l, \ |N(S)| \geq |S| - c(N(S)).\]
Hence, in graph $H[A, G] - \{i, j\}$ we have
\[ \forall S \subseteq G_1, \ |N(S) \cap A_1| \geq |S| - c(N(S)).\]

The lemma follows.
\end{proof}
\end{document}